\documentclass[aos]{imsartclean}

\RequirePackage{amsthm,amsmath,amsfonts,amssymb}
\RequirePackage[numbers,sort&compress]{natbib}
\RequirePackage[colorlinks,citecolor=blue,urlcolor=blue]{hyperref}
\RequirePackage{graphicx}
\usepackage{subcaption}
\usepackage{upgreek}
\usepackage{amsmath}
\usepackage{mathrsfs}
\usepackage{enumerate}
\usepackage{xcolor}
\usepackage{booktabs}
\usepackage[font=footnotesize]{caption}  
\newcommand{\tr}{\mathrm{tr}}
\newcommand{\cov}{\mathrm{cov}}
\input{tcilatex}
\startlocaldefs
\theoremstyle{plain}

\theoremstyle{definition}

\newtheorem {assumption}{Assumption}
\numberwithin {equation}{section} 
\numberwithin {theorem}{section} 
\endlocaldefs

\begin{document}

\begin{frontmatter}
\title{Optimal Break Tests for Large Linear Time Series Models}
\runtitle{Optimal Break Tests}

\begin{aug}
\author[A]{\fnms{Abhimanyu}~\snm{Gupta} \ead[label=e1]{abhimanyu.g@queensu.ca}} \and
\author[B]{\fnms{Myung Hwan}~\snm{Seo}\ead[label=e2]{myunghseo@snu.ac.kr}}

\address[A]{Department of Economics, Queen's University, Dunning Hall, 94 University Avenue, Kingston, Ontario, K7L 3N6, Canada\printead[presep={ ,\ }]{e1}}

\address[B]{Department of Economics, Seoul National University, Gwan-Ak Ro 1, Seoul, Korea\printead[presep={,\ }]{e2}}
\end{aug}

\begin{abstract}
We develop a class of optimal tests for a structural break occurring at an unknown date in infinite and growing-order time series regression models, such as AR($\infty$), linear regression with increasingly many covariates, and nonparametric regression. Under an auxiliary i.i.d. Gaussian error assumption, we derive an average power optimal test, establishing a growing-dimensional analog of the exponential tests of \cite{Andrews1994a} to handle identification failure under the null hypothesis of no break. Relaxing the i.i.d. Gaussian assumption to a more general dependence structure, we establish a functional central limit theorem for the underlying stochastic processes, which features an extra high-order serial dependence term due to the growing dimension. We robustify our test both against this term  and finite sample bias and illustrate its excellent performance and practical relevance in a Monte Carlo study and a real data empirical example. 
\end{abstract}



\end{frontmatter}

\section{Introduction}
\label{sec:intro}

This paper uses a growing-dimension asymptotics approach to develop a class of asymptotically valid tests for the stability of infinite or growing dimensional linear regression models, which we will call `large linear models', when the change point is unknown. These tests are motivated by optimality ideas explored in the parametric literature and are indeed shown to possess analogous optimal properties under certain conditions. Our terminology of infinite-dimensional models includes both autoregressive models of infinite-order and nonparametric regression. This paper is thus a contribution to the recently growing interest in making inference on a parameter vector whose dimension is not finite.

Growing-dimension asymptotics have been studied in the statistical literature since at least the important work of \cite{huber1973robust}, but the sequences of experiments in \cite{LeCam1960} constitute an even earlier reference. Our growing-dimensional approach is general enough to include many important models as special cases. For example, we cover the autoregressive model of infinite-order, AR($\infty$), one of the most important models in time series analysis. Since a nonparametric regression function can be represented as an infinite-order linear regression in a linear sieve space, our methods cover this case as well. We allow for a triangular array structure and can thus also accommodate linear regression with growing dimension, but being explicit about the approximation error allows the coverage of true infinite dimensional models such as AR($\infty$) and nonparametric regression. Our focus on linear regression facilitates clarity of exposition, although extensions are feasible to more general settings like nonlinear regression, maximum likelihood, or generalized method of moments (GMM). 

While we focus on testing for the presence of a structural break in the AR($\infty$) model, nonparametric regression and growing dimension linear regression,  we note that there is a thriving literature on the estimation of high-dimensional change point models, often with LASSO-type methods, 
e.g. \cite{Chan2014}; \cite{Li2016}; \cite{lee2016lasso}; \cite{wang2018high}; \cite{lee2018oracle}; \cite{Safikhani2022}; \cite{Wang2024}. This topic has also attracted much attention also in the fixed-dimensional parametric setting, e.g. \cite{bai1998estimating}, among many others. 
Our framework is also different from another closely related literature, which explores testing of stability of means or variance-covariances of multiple time series, e.g. \cite{aue2009break, cho2015multiple, wang2022inference, li2024detection}. In a parametric setting, \cite{Andrews1993} considers testing under a GMM framework, which has been extended to the case of a break occurring anywhere in the sample by \cite{hidalgo2013testing}, and to the systems of equations setting by \cite{qu2007estimating}. \cite{Andrews1994a} and \cite{elliott2006efficient} explore certain optimality properties which form the basis for the tests considered in this paper.

More generally, testing for structural stability is a heavily studied subject and a huge literature has evolved, as reviewed by \cite{perron2006dealing,aue2013structural}. In the infinite and growing dimensional-setting, however, the literature is not as rich. \cite{Wu1993} consider testing for jumps in a regression function using kernels and under a fixed design, \cite{vogt2015testing} studied break testing in nonparametric kernel regression, and \cite{Mohr2019} consider testing in a time series regression model using marked empirical processes of residuals, while \cite{Hidalgo2016} use a non-smoothing approach to test the stability of a nonparametric regression function.

To explain our approach further, we view our methodology as approximating an infinite-dimensional testing problem with a sequence of finite-dimensional problems. In case of the
nonparametric regression or AR($\infty$), this is based on sieve or series estimation which entails approximating the infinite dimensional primitive with a linear combination of increasingly many basis functions. On the other hand, if the model of interest is an increasing dimensional one, the approach is interpreted as inference on a sequence of increasingly complex models. Because the sequence of basis functions (nonparametric regression), lags (AR($\infty$)) or covariates (increasing dimension regression) that underlies our approach is not of fixed dimension,
one cannot simple apply the theory of finite-dimensional parametric
problems. Instead, the `usual' parametric test statistics must be adjusted to take into account the increasing dimensional nature of the problem, as shown by \cite{Portnoy1988}. \cite{Murphy1993} used a similar approach in a proportional hazards model while \cite%
{Hong1995} propose an appropriately standardized test statistic that provides a consistent specification test based on series estimation. \cite{Fan2001} provides a Wilks Phenonemon type result based on a standardized statistic. 


Our testing problem also features a loss of identification under the null of no structural break.
Thus, we consider a stochastic process of recentered and renormalized quadratic forms indexed by the unidentified parameter that captures the change point as a fraction of the sample size. We establish the weak convergence of this stochastic process as both the sample size $T$ and the number $p $ of restrictions grow to infinity simultaneously and construct test statistics based on a weighted exponential transformation motivated by the optimal tests in \cite{Andrews1994a}. Under an i.i.d. Gaussian design, we establish a weighted power optimality for our class of tests in the same spirit.  

Optimality typically requires a finite
number of directions in which the null may be violated, see e.g. Section
14.6 of \cite{Lehmann2006}. To account for the infinite-dimensional aspect
of the problem, we adopt the spirit of sieve regression and
assume that we can direct our test to have reasonable power against an
increasing number of approximating directions, while the remaining `tail'
directions can be neglected. Among the approximating directions, we assign
weights inspired by \cite{Andrews1994a}. In parametric structural break
testing, \cite{elliott2006efficient} found that most optimal tests are
similar as long as they are optimal against alternatives of similar average
deviation from the null. When $p$ grows to infinity, small changes in each coefficient
can aggregate, so that individually negligible changes can add up to a non-negligible change as a whole. 
Indeed, we show that the determinant of our local power is the average of all the changes.

In the general model without the i.i.d. Gaussianity assumption, we observe that an important feature of the limiting process is the presence of a \emph{nonlinear high order long run variance} (HLV)  involving the errors and regressors in a nonlinear fashion, \cite{Gupta2023}. 
We also show that the sequential limit fails to uncover this nonlinear serial correlation, where 
the limits are taken sequentially first as $T\rightarrow \infty $ and then as $p\rightarrow \infty .$ 
We robustify our test against this HLV factor using a random scaling method as opposed to resorting to consistent estimation. This yields a pivotal test and is in the spirit of the self-normalization approach, see e.g. \cite{shao2010testing, Zhang2018}. In addition, we incorporate a bootstrap bias correction that controls size effectively in finite samples while preserving power.

We also examine the finite-sample performance of our test in a simulation study. We find that our test exhibits very good size control and high power against a variety of alternatives, both in nonparametric regression and long AR estimation. Importantly, we find that failure to correct for the nonlinear serial correlation that we have discovered can lead to seriously distorted sizes. We show that this distortion can be particularly severe as $p,T$ increase to infinity, and the correction we propose dramatically improves matters in many cases. In an empirical illustration with real data we examine the relationship between oil prices and economic output and find that our corrections can yield different conclusions to existing methods.  

The next section introduces the class of optimal tests that motivate our approach, sets out the \emph{mise en sc\`ene} of our structural break model and establishes basic asymptotic results and the optimality property of our tests. In Section \ref{sec:exp_rob_piv} we show how to construct HLV-robust and bias corrected versions of our tests and establish their asymptotic properties, while Section \ref{sec:sequential} shows that taking a sequential limit approach would hide the HLV factor that our simultaneous limit theory reveals. In Section \ref{sec:mc} we demonstrate the desirable finite sample properties of our tests in a set of Monte Carlo simulations. Finally, we illustrate our tests on real data in Section \ref{sec:application}.

\section{Optimal Testing}\label{sec:optimal}

We develop an optimal test for the presence of a structural break in infinite-order regression such as the conditional mean of nonparametric regression models with iid Gaussian errors in terms of a weighted average power. The growing dimensional regression model is a special case of our model. It generalizes conventional optimal structural break testing, which \cite{Andrews1994a} proposed for finite-dimensional regression. Specifically, we first characterize the optimal test under:
\begin{equation}
	y_{Tt}=\boldsymbol{x}_{t}^{\prime }\gamma _{0}+\left(T\sqrt{p}\right)
	^{-1/2}\psi_{Tt} 
	\sum_{j=1}^{\infty }\boldsymbol{x}_{tj}b_{Tj}+\sigma_0 \nu _{t},
	\label{eq:class for optimality}
\end{equation}
which defines a sequence of models where $\psi _{Tt}=\psi -1\left\{ t/T\leq \psi \right\} $ for  $\psi \in \Psi=\left[c_1,c_2\right] \subset (0,1)$,  $\{\nu _{t}\}$ is an i.i.d. sequence of standard normals and is independent of the time series of infinite-dimensional vectors $\left\{ \boldsymbol{x}_{t}\right\} $, while $p$ grows with the sample size $T$.  We drop the subscript $T$ to ease notation hereafter unless it is necessary to avoid confusion. The null hypothesis of interest is given by  
$$H_0 : (b_{1},b_{2},...)'=0, \quad \text{for any } \psi \in \Psi.  $$ 

For given weight functions for $\psi$ and $\boldsymbol{b}= (b_{1},b_{2},...)'$, which we denote by $J\left( \psi \right) $ and $Q\left( \boldsymbol{b}\right) $ respectively, the weighted average power of a test $\mathscr{T}_{T}$ over this class of local alternative models is written as 
\begin{equation*}
	\int \int \int \mathscr{T}_{T}p_{\boldsymbol{b},\psi }\left( Y|\boldsymbol{X}\right)
	dydQ\left( \boldsymbol{b}\right) dJ\left( \psi \right) =\int \mathscr{T}_{T}\int
	\int p_{\boldsymbol{b},\psi }\left( Y|\boldsymbol{X}\right) dQ\left( 
	\boldsymbol{b}\right) dJ\left( \psi \right) dy,
\end{equation*}
where $Y$, $ \boldsymbol{X}$, and $p_{\boldsymbol{b},\psi }\left(Y|\boldsymbol{X}\right)$ denote the collections of $\left\{
y_{t}\right\} $, $\left\{ \boldsymbol{x}_{t}^{\prime }\right\} ,$ and the conditional density, respectively, and the equality holds by Fubini's theorem.  The final expression reveals that an optimal test can be constructed so as to maximize the weighted average likelihood $\int \int
p_{\boldsymbol{b},\psi }\left( Y|\boldsymbol{X}\right) dQ\left( 
\boldsymbol{b}\right) dJ\left( \psi \right) $ for any given probability measures $Q$ and $J$. In this alternative representation, the alternative hypothesis is a simple hypothesis. That is, the alternative hypotheses can be written as
\[H_{A}:\int \int \phi \left( \boldsymbol{x}_{t}^{\prime }\gamma _{0}+\left( T\sqrt{p}\right) ^{-1/2}\psi _{t}\boldsymbol{x}_{t}^{\prime }\boldsymbol{b},\sigma_{0}\right) dQ\left( \boldsymbol{b}\right) dJ\left( \psi \right) ,
\]
where $\phi \left( \mu ,\sigma \right) $ denotes the normal density function with mean $\mu $ and variance $\sigma ^{2}$. Assuming $\gamma _{0}=0$ (without loss of generality) yields $y_{t}=\sigma _{0}\nu _{t}$ under the null. Then, the likelihood ratio (LR) statistic is given by 
\begin{equation*}
	LR_{T}=\int \int \func{exp}\left( -\frac{1}{2\sigma _{0}^{2}}\sum_{t=1}^{T}\left(
	\left( y_{t}-\left( T\sqrt{p}\right) ^{-1/2}\psi _{t}\boldsymbol{x}%
	_{t}^{\prime }\boldsymbol{b}\right) ^{2}-y_{t}^{2}\right) \right) dQ\left( 
	\boldsymbol{b}\right) dJ\left( \psi \right).
\end{equation*}
As noted by \cite{Andrews1994a}, this is a useful feature since optimality can be established for the likelihood ratio test by virtue of the Neyman-Pearson lemma.

Thus, it is important to specify the weights on $\boldsymbol{b}$ and $%
\psi $ in an appropriate fashion. A particular choice of the weighting function was proposed by \cite{Wald1943} for a fixed $\psi $ and extended by \cite{Andrews1994a} and \cite{Song2009} for the case where $\psi $ is not identified. The latter papers consider a finite number of restrictions on functions of $\psi $ and the weight function assigns constant weight on the ellipses in the parameter space of $ \boldsymbol{b} $ that are equally difficult to detect for each $\psi$. As noted by \cite{Andrews1994a}, every locally most powerful invariant test has to exhibit a constant power over these ellipses under the parameterization in \eqref{eq:class for optimality}. Our testing problem concerns an infinite-dimensional restriction that is also a function of $\psi$.  

\sloppy Building on the observation by \cite{Lehmann2006} (Section 14.6) that one cannot obtain reasonable local power across all families of distributions when the restriction concerns an infinite-dimensional parameter, we consider weighting functions that coincide with these ellipses only up to a truncation parameter, which is set as $p$ in \eqref{eq:class for optimality}.  
Specifically,  let $x_{t}$ stand for the first $p$-elements of $\boldsymbol{x}_{t}$ and let $w_{t}$ denote the remaining elements so that $
\boldsymbol{x}_{t}=(x_{t}^{\prime },w_{t}^{\prime })^{\prime }$ and partition $\boldsymbol{b}=\left( \boldsymbol{b}_{1}^{\prime },\boldsymbol{b}
_{2}^{\prime }\right) ^{\prime }$ accordingly. Recall that we drop $T$ in the notation.
Also, write
$
Q\left( \boldsymbol{b}\right) =Q_{1}\left( \boldsymbol{b}_{1}\right)
Q_{2}\left( \boldsymbol{b}_{2}|\boldsymbol{b}_{1}\right) ,
$
where $Q_{1}\left( \boldsymbol{b}_{1}\right) $ is the $p$-dimensional centered multivariate normal distribution function with covariance matrix $%
\mathfrak{c}\sigma_0^{2}\left( \psi \left( 1-\psi \right)
E{x}_{t}{x}_{t}^{\prime }\right) ^{-1}$ for some finite  $\mathfrak{c}$ and $Q_{2}\left( \boldsymbol{b}_{2}|%
\boldsymbol{b}_{1}\right) $ is a conditional distribution  over the support 
$
\left\{ \boldsymbol{b}_{2}:\left\Vert \boldsymbol{b}_{2}\right\Vert _{2}\leq
K\ \text{and }\sqrt{p}\left\Vert \boldsymbol{b}_{2}\right\Vert _{2}=o\left(
1\right) \right\},
$
for $\boldsymbol{b}_{2}$ given $\boldsymbol{b}_{1}$ and some finite $K$. The growing dimensional model sets $\boldsymbol{b}_2 = 0$. The weight $J$ for $\psi $ may be set as uniform on the interval $\Psi $ unless there is a specific set of break points of interest.

The restriction on the decay rate of the infinite-dimensional parameter $\boldsymbol{b}_{2}$  arises naturally from the features of the particular model under consideration. For instance, \cite{berk1974consistent} and \cite{Kuersteiner2005} constrain the autoregressive coefficients to meet the decay rate of $\sum_{j=1}^{\infty
}j\left\vert \gamma _{j}\right\vert <\infty $ for the stationarity and moment
conditions of the process and consistent estimation of the spectral density at the origin. On the other hand, in the nonparametric regression context, the smoothness condition on the regression functions can be imposed by constraining the decay rate of the sieve coefficients. 

To illustrate, consider the Fourier representation $f\left( x\right) =\sum_{j=1}^{\infty }\gamma _{j}\exp
\left( ijx\right) $, for which the $m$-th derivative is given by $f^{\left(
	m\right) }\left( x\right) =\sum_{j=1}^{\infty }\gamma _{j}\left( ij\right)
^{m}\exp \left( ijx\right). $ Thus, it is natural to assume that $%
\sum_{j=1}^{\infty }\left\vert \gamma _{j}\right\vert j^{m}<\infty $ for $%
f^{\left( m\right) }\left( X\right) $ to be well-defined \emph{a.s.} This
also implies that $\gamma _{j}=o\left( j^{-m-1}\right) $. These observations are reflected in our construction of the weight function. The preceding specification of local models with truncation $p$ enables us to construct a class of computationally tractable tests that are asymptotically equivalent to the LR test. We develop such an asymptotically optimal test in the following subsection, extending the finite-dimensional parametric optimality results of \cite%
{Andrews1994a}.
\subsection{Structural Break in Growing and Infinite-Dimensional Regression}
\label{sec:model} 
We consider the linear regression model of growing dimension with a potential structural break at $ t=[T \psi] $,
\begin{equation}
	y_{t} 
	=x_{t}^{\prime }\kappa _{1}+x_{t}^{\prime }\kappa _{2}1\left\{ t/T>\psi
	\right\} +u_{t}.\label{our_strucb_model}
\end{equation}
where $\psi\in\Psi$ is unknown, $x_t$ is a $p$-dimensional observable vector, $\kappa_1, \kappa_2$ are $p$-dimensional unknown parameters, and $p\rightarrow \infty $ as $T\rightarrow \infty $. This is the truncated model we employ to construct our test statistic and $\kappa_2$ corresponds to $\boldsymbol{b}_1$ so that we construct a test for the null hypothesis 
\[
\mathcal{H}_0 : \kappa_2 =0 \quad \text{for any } \psi \in \Psi. 
\]

We will allow for general increasing dimension and nonparametric time series regressions, $E\left( y_{t}|%
\mathfrak{G}_{t-1}\right) $, where $\mathfrak{G}_{t-1}$ denotes the filtration up to time $t-1$, thus generalizing the Gaussian model \eqref{eq:class for optimality}. 
Let  $\nu _{t}=y_{t}-E\left( y_{t}|\mathfrak{G}_{t-1}\right) $, $\gamma_T$
be the best linear predictor of $y_{t}$ given $x_{t}$, and $r_{Tt}=E\left(
y_{t}|\mathfrak{G}_{t-1}\right) -x_{Tt}^{\prime }\gamma_T $.  Then, $u_{t}=r_{Tt}+\nu _{t}$, i.e. the error term comprises of the regression error in addition to the approximation error. Here we use the subscript $T$  to emphasize the the array structure temporarily. 

For example, for growing dimensional regression the model features an increasing number of observable covariates $x_{t}$ as the sample size diverges, a modelling approach that dates back to the pioneering work of \cite{huber1973robust}. In the case of infinite order autoregression, we include the lagged dependent variables, $\left\{y_{t-j}\right\}_{j\geq 1}$, and our approach indicates a truncation at $p$ lags, $p\rightarrow\infty$, as suggested in the classical contributions of \cite{berk1974consistent} and \cite{Shibata1980}, among others. Finally, in the case of nonparametric series regression, the model features a finite
	number of observable covariates $z_{t}$ and
	$
	x_{Tt}:=x_{T}\left( z_{t}\right) :{\mathbb{R}}^{k}\mapsto {\mathbb{R}}^{p}$ are transformations via known basis functions.

Let $x_{t}\left( \psi \right) :=\left(x_{t}^{\prime },x_{t}^{\prime }1\left\{ t/T>\psi \right\} \right) ^{\prime}$ and denote by $ \hat{\kappa }\left( \psi \right) $ and $ \hat{u}_t\left( \psi \right) $ the OLS estimate and the OLS residuals, respectively, from the linear regression with a fixed $\psi$.  
Define
$ 
\hat{P}\left( \psi \right) ={T}^{-1}\sum_{t=1}^{T}x_{t}\left( \psi \right) x_{t}\left( \psi \right) ^{\prime },
$
and let $\hat{\Xi}\left( \psi \right) $ denote an estimator of $%
E\nu _{t}^{2}x_{t}\left( \psi \right) x_{t}^{\prime }\left( \psi
\right) $. One can use $T^{-1}\sum_{t=1}^{T}x_{t}\left( \psi \right) x_{t}\left(
\psi \right) ^{\prime }\hat{u}_{t}\left( \psi \right) ^{2}$ (the Eicker-White formula) or  even
$\hat{\sigma}\left( \psi \right) ^{2}%
\hat{P}\left( \psi \right) $, where $\hat{\sigma}^{2}\left( \psi \right)
=T^{-1}\sum_{t=1}^{T}\hat{u}_{t}\left( \psi \right) ^{2}$, if conditional homoskedasticity obtains. Define the heteroskedasticity robust Wald process as
\begin{equation}
W_{T}\left( \psi \right) :=T\hat{\kappa}_{2}\left( \psi \right) ^{\prime
}\left( S\hat{P}\left( \psi \right) ^{-1}\hat{\Xi}\left( \psi \right) 
\hat{P}\left( \psi \right) ^{-1}S^{\prime }\right) ^{-1}\hat{\kappa}%
_{2}\left( \psi \right) , \quad \psi\in\Psi,  \label{wald_def}
\end{equation}%
where $S=\left( 0_{p\times p}:I_{p}\right) $, with $I_p$ the $p\times p$ identity matrix.
To account for the growth of $ p $ as a function of $ T $, we center and rescale $W_T(\psi)$  as
\begin{equation}
{\mathcal{Z}}_{T}\left (\psi \right ):=\left({W_{T}\left (\psi \right )-p%
}\right)/{\sqrt{2p}}, \quad \psi\in\Psi. \label{Q_n}
\end{equation}
The idea is that this standardization yields a proper limit as $p$ and $T$ diverge to infinity simultaneously, motivated by the intuition that a correctly centered and scaled sequence of chi-squared random variables converges to a normal variate in distribution. It is motivated by \cite{Portnoy1988, DeJong1994, Hong1995, Fan2001, Gupta2023}, for example. They established a CLT for the centered and rescaled statistic $\mathcal{Z}_{T}\left (\psi \right )$ for the case where $\psi$ is fixed at a constant, under various sampling conditions. 

We extend the CLT for time series data in \cite{Gupta2023} to a functional CLT (FCLT) by establishing the weak convergence of  ${\mathcal{Z}}_{T}\left (\psi \right )$ as a process indexed by $\psi\in \Psi$ in order to construct an optimal test for the structural break with an unknown break date.
Define the process 
$$
	\mathcal{Z}(\psi )={\psi^{-1}G(\psi )+(1-\psi)^{-1}\bar { G}\left (\psi \right )-G(1)},$$
where $\left (G\left (\psi \right ),\bar {G}\left (\psi
\right
)\right
)^{\prime }$ is a bivariate Gaussian process with
covariance kernel 
\begin{equation}  \label{C_def}
	{\mathcal{K}}\left (\psi _{1},\psi _{2}\right )=\left ( 
	\begin{array}{cc}
		\left (\min\left\{\psi _{1}, \psi _{2}\right\}\right )^{2} & 1\left \{\psi
		_{1}>\psi _{2}\right \}\left (\psi _{1}-\psi _{2}\right )^{2} \\ 
		1\left \{\psi _{1}<\psi _{2}\right \}\left (\psi _{1}-\psi
		_{2}\right )^{2} & \left (1-\max\left\{\psi _{1}, \psi _{2}\right\}\right
		)^{2}%
	\end{array}
	\right ).
\end{equation}
Observe that the marginal distribution of $ \mathcal{Z}(\psi ) $ is  standard normal.

\sloppy Under the presence of strong high-order serial dependence, the limit of the finite-dimensional distribution of $ \mathcal{Z}_T (\psi)$ features the HLV factor as shown by \cite{Gupta2023}. For convenience, we let $\varsigma_{t}=\Xi^{-1/2}x_{t}\nu_{t}$, where $\Xi$ is defined in Assumption \ref{ass:M_diff}, $\delta_t(\psi)=\left({1}\left\{t/n\leq \psi\right\}-\psi\right)/\sqrt{\psi(1-\psi)}$ and define the HLV factor as
\begin{equation} 
	\upomega=2\lim_{T\rightarrow\infty}\frac{1}{T}\sum_{s,t=2}^T cov \left(q_s(\psi),q_t(\psi)\right),
	\label{V_def1}
\end{equation}%
where $q_t(\psi)=(Tp)^{-1/2}\delta_t(\psi)\varsigma_t'\sum_{s<t}\delta_s(\psi)\varsigma_s$, and any limit stated as `$T\rightarrow\infty$' is taken as both $ T $ and $ p $ grow to infinity simultaneously. A sequential limit is considered in Section \ref{sec:sequential}. 
Since the variable $q_t(\psi)$ is a  nonlinear transformation of the underlying primitive random variables $x_t$ and $\nu_t$ and their entire past, its variance depends on $t$ and contains higher-order autocovariances of the primitive variables or $\varsigma_t$. Higher-order dependence of a martingale difference sequence (mds) process is a common feature of many time series as in e.g. the ARCH process. In this sense it differs from the familiar long-run variance encountered in fixed dimensional time series regressions, which is the limiting sum of all the autocovariances.

Throughout let $C$ denote
a generic finite constant, independent of $T$. 
\begin{assumption}
	\label{ass:errors} The martingale difference sequence $\left\{ \nu
	_{t}\right\} $ satisfies  $\sigma _{t}^{2}\leq C,$ where $E\left(
	\nu _{t}^{2}|\mathfrak{G}_{t-1}\right) =\sigma _{t}^{2},$ and $%
	E\left( \nu _{t}^{4}|\mathfrak{G}_{t-1}\right) \leq C$, a.s.
\end{assumption}


\begin{assumption}
	\label{ass:aprx0} For $a=1,2$, $
		\sup_{t}E\left( r_{Tt}^{2a}\right) =o\left( T^{-1}\right) . $
\end{assumption}

Introduce the $p\times p$ non-random matrix sequences $P$ and $ \Xi $ that satisfy Assumption \ref{ass:M_diff} below and let 
\begin{equation*}
	P(\psi )=\left[ 
	\begin{array}{cc}
		P & (1-\psi )P \\ 
		(1-\psi )P & (1-\psi )P%
	\end{array}%
	\right] ,\;\;\Xi (\psi )=\left[ 
	\begin{array}{cc}
		\Xi & (1-\psi )\Xi \\ 
		(1-\psi )\Xi & (1-\psi )\Xi%
	\end{array}%
	\right] .
\end{equation*}%
Define $\left\Vert A\right\Vert =\left\{ \overline{\text{eig} }(A^{\prime
}A)\right\} ^{\frac{1}{2}}$ for a generic matrix $A$, where $\underline{%
	\text{eig} }$ (respectively $\overline{\text{eig} }$) denotes the smallest
(largest) eigenvalue of a symmetric nonnegative definite matrix. 
\begin{assumption}
	\label{ass:M_diff} 
	(i) $\sup_{i,t}Ex_{ti}^{4}<\infty $ 
    
    (ii) 
    \[
    \sup_{r\in\Psi\cup\{1\}}\left(
		\left\Vert T^{-1}\sum_{t=1}^{[Tr]} x_{t}x_{t}^{\prime }-rP\right\Vert +
		\left\Vert T^{-1}\sum_{t=1}^{[Tr]}x_{t}x_{t}^{\prime }\sigma _{t}^{2}-r \Xi \right\Vert\right)
		=O_{p}\left(
		\tau _{p}\right),
        \]
			\[
            \sup_{r\in\Psi\cup\{1\}}
		\left\Vert \hat{\Xi}\left( r \right)
		-\Xi \left( r \right) \right\Vert =O_{p}\left( \zeta_{p}\right) ,
        \]
    and	$\underline {\text{eig} }
		\left ({P}\right )> \mu_T,\underline {\text{eig} }
		\left ({\Xi}\right )> \mu_T,$
	for some positive,  $ \tau_{p} $, $\zeta_p$ and $\mu_T$, satisfying
	\begin{equation}
		\mu_T^{-4} \sqrt{p}\left( \mu_T^{-1}\tau _{p}+\zeta_{p}\right)+\mu_T^{-6} p^{-1} \rightarrow 0, \text{ as } T\rightarrow \infty  .  \label{rate:Q_weak_conv}
	\end{equation}
	(iii)  $\varlimsup _{T\rightarrow \infty }\overline {%
		\text{eig} }\left (P\right )<\infty $, 
	$\varlimsup _{T\rightarrow \infty }%
	\overline {\text{eig} }\left (\Xi \right )<\infty $.	
\end{assumption}
Primitive conditions and expressions for $\tau _{p}$ and $\zeta_p$ are given in Propositions B1 and B2 of \cite{Gupta2023}. Recall that the eigenvalues of the Kronecker product of two symmetric matrices are the products of their eigenvalues, and $ \psi $ is bounded away from zero and one. Thus, $P(\psi)$ and $\Xi(\psi)$ inherit the eigenvalue restrictions on $P$ and $\Xi$ in Assumption \ref{ass:M_diff} $(ii)$ and $(iii)$, up to positive constants.

Let $ \mathfrak{H}_t $ denote a filtration for $ \varsigma_{t} $, $\Phi_t=E\left(\varsigma_{t} \varsigma_{t}'|\mathfrak{H}_{t-1}\right)$,  and $\Theta_s=\sum_{t_1=1}^{s-1}\sum_{t_2=1}^{s-1} \varsigma_{t_1} \varsigma_{t_2}'$. The filtration $ \mathfrak{H}_t $ need not be $ \mathfrak{G}_t $ but a simpler one as long as it makes $ \varsigma_t $ an mds.  Indeed, some conditions may be easier to verify  under simpler filtrations. 
The next assumption introduces the HLV factor $\upomega$ formally.

\begin{assumption}\label{ass:MCLT}
 $T^{-\varphi}\max_{1\leq t \leq T}\overline{\text{eig}}\left(\Phi_t\right)+T^{-\lambda}\max_{1\leq t\leq T} E\left(\left(\varsigma_t '\varsigma_t\right)^2|\mathfrak{H}_{t-1}\right)=o_p (1) $, for some $\varphi\in[0,1/3)$ and some $\lambda\in[0,1-\varphi)$, $\sum_{t=1}^T\sum_{s=1}^{t-1}\cov\left(\tr\left(\Phi_{t}\Theta_{t}\right),\tr\left(\Phi_{s}\Theta_{s}\right)\right)=o(T^{4}p^{2})$,
and there exists $ \upomega $ such that for $ m $ that is proportional to $ T $ and $ l=0 $ or $ [T\psi] $
\begin{equation} 
		\lim_{T \rightarrow \infty } \frac{1}{mp}tr\sum_{t_{1}=1}^{m-1}%
		\sum_{t_{2}=1}^{m-1 }E\left( \varsigma_{m+l}\varsigma_{m+l}^{\prime
		}\varsigma_{t_{1}+l}\varsigma_{t_{2}+l}^{\prime }\right)  
		= \upomega,
		\label{V_partial_def}
	\end{equation}
uniformly in $\psi\in\Psi$.
\end{assumption}
As we discussed subsequently to \eqref{V_def1}, this quantity in the asymptotic distribution differs from the conventional long-run variance and consists of fourth-order cross-moments of a time series. Then, a resampling algorithm matching only up to second order moments such as the wild bootstrap would not be valid. Also note that if the $v_t$ are i.i.d. Gaussian then $\upomega=1$. 

For mean zero random variables $a_{1i},a_{2j},a_{3k},a_{4l}$, let $\mathrm{cum}_{ijkl}\left(a_{1i},a_{2j},a_{3k},a_{4l}\right)$ denote the fourth cumulant. 
\begin{assumption}\label{ass:autocovandcumulant}
	$\left\{x_{ti}\nu_t\right\}_{t\in\mathbb{Z}}$ is fourth order stationary for all $i=1,\ldots,p$. Furthermore, $\sup_{i,j=1,\ldots,p}\sum_{t=-\infty}^{\infty}\left\vert \upgamma_{ij}(t)\right\vert<\infty$, where $\upgamma_{ij}(t)=E\left(x_{r,i}\nu_{r}x_{r+t,j}\nu_{r+t}\right)$ for integer $r$, and $\sup_{i,j,k,l=1,\ldots,p}\sum_{t_1,t_2,t_3=-T}^T \left\vert\mathrm{cum}_{ijkl}\left(x_{0,i}\nu_{0},x_{t_1,j}\nu_{t_1},x_{t_2,k}\nu_{t_2},x_{t_3,l}\nu_{t_3}\right)\right\vert=O\left(T^2\right)$.
\end{assumption}
This assumption controls the temporal dependence in $\left\{x_t\nu_t\right\}$ and is discussed in \cite{Andrews1991a}, for example, wherein sufficient conditions for it to hold are also provided. The next two theorems establish the stochastic equicontinuity of the process $\mathcal{Z}_T(\psi)$ and its weak convergence.
Let `$\Rightarrow$' denote weak convergence in $\ell^{\infty}(\Psi)$. 
\begin{theorem}
	\label{thm:null_Chow} Let Assumptions \ref{ass:errors}- \ref{ass:autocovandcumulant} 
	and $\mathcal{H}_{0}$ hold.
	Then $
	{\mathcal{Z}}_{T}(\psi )\Rightarrow \sqrt{\upomega} \mathcal{Z}(\psi)$. 
\end{theorem}

For local power properties, let the sequence of local alternatives be
\begin{equation}
	\mathcal{H}_{\ell }:\kappa _{2\ell }=2^{1/4}\varrho p^{1/4}/\sqrt{T},
	\label{local_alternatives}
\end{equation}%
where $\varrho $ is a unit length $p\times 1$ vector. Due to this unit length condition, the power of $p$ becomes $1/4$ from $-1/4$ in \eqref{eq:class for optimality}.  This sequence converges to the null slower than the usual $1/\sqrt{T}$ parametric rate due to the inflation caused by the $p^{1/4}$ factor. This slower rate is the price to pay for a nonparametric approach. 
\begin{theorem}
	\label{thm:local_power} \sloppy Suppose that Assumptions \ref{ass:errors}- \ref{ass:autocovandcumulant} and $\mathcal{H}_{\ell}$ hold and let $\varrho _{\infty }=\lim_{T\rightarrow \infty }\varrho ^{\prime
	}P\Xi ^{-1}P\varrho $. Then, 
	$
	{\mathcal{Z}}_{T}(\psi )\Rightarrow  \sqrt{\upomega}\mathcal{Z}(\psi)+\varrho _{\infty } {\psi (1-\psi )}
	$.
    
\end{theorem}


\subsection{Asymptotic Optimality}

To construct our test statistics, we exploit an approach to testing that was initiated by \cite{Andrews1994a} for parametric models and has been widely applied e.g. \cite{Song2009}.  
This involves constructing test statistics by weighted exponential transformations of an underlying process. These are motivated by the fact that, under certain simplifying assumptions, tests constructed in this fashion have an optimality property, as we explore in Section \ref{sec:optimal}.
In the parametric setting, \cite{elliott2006efficient} showed that many of the  seemingly different optimal tests are in fact asymptotically equivalent
within a large class of local break models and what matters is the weighted  average size of the breaks, regardless of the number of breaks. These references do not treat the stability of nonparametric or growing dimensional models. 


Accordingly, define the following class of test statistics:
\begin{equation}
	Exp{\mathcal{Z}}_{T}(\mathfrak{c})=\frac{\sqrt{2}}{\mathfrak{c}}\log \int_{\Psi }\func{%
		exp}\left( \frac{\mathfrak{c}}{\sqrt{2}}\mathcal{Z}_{T} \left( \psi \right) \right) dJ\left( \psi \right)  \label{eq:expQn1},
\end{equation}
for a positive $\mathfrak{c}$ and a bounded weight function $J(\cdot)$ such that $\int_{\Psi }dJ\left( \psi \right) =1$.  These are \cite{Andrews1994a}-style test statistics but correctly constructed to account for $p\rightarrow\infty$ as well as the centered and rescaled Wald process $\mathcal{Z}_{T} \left( \psi \right)$ that we introduced above. 

The optimality of the $Exp{\mathcal{Z}}_{T}$ test under the Gaussian design is
formally given by the following equivalence theorem. We will elaborate on the robustification against the HLV factor in the subsequent section.
\begin{assumption}
	\label{ass:opt1} $Ex_{t}^{\prime }x_{t}=O\left( p\right) $, the minimum
	eigenvalue $\underline{\text{eig} }\left(Ex_{t}x_{t}^{\prime }\right)$ is bounded away
	from zero, and $\max_{t}\sup_{\left\vert a\right\vert _{2}=1}E\exp \left(
	\left\vert w_{t}^{\prime }a\right\vert ^{2}\right) <\infty $.
\end{assumption}
\begin{theorem}
	\label{theorem:optimality} Let the specification in (\ref{eq:class for
		optimality}) hold together with Assumptions \ref{ass:errors}-\ref{ass:opt1}. Then, 
	$$
	LR_{T}=e^{-\mathfrak{c}^{2}/4}Exp{\mathcal{Z}}_{T}+o_{p}(1),
	$$
	under both $H_{0}$ and $H_{\ell }$.
\end{theorem}

To shed light on our optimality, we elaborate the relation of our $Exp{\mathcal{Z}}_{T}$ test statistic to the optimal test statistic Exp-W introduced by \cite{Andrews1994a} for the parametric model, viz. 
$\left( 1+ \mathfrak{c}\right) ^{-\frac{p}{2}}\int_{\Psi }\func{exp}\left( \frac{1}{2} {\mathfrak{c}}/({1+\mathfrak{c})}W_{T}\left( \psi \right) \right) dJ\left( \psi \right) $, where $W_{T}\left( \psi \right) $ is the Wald statistic defined in (\ref{wald_def}). Specifically, \eqref{eq:expw} in the proof of  Theorem \ref%
{theorem:optimality}  shows that 
\begin{eqnarray}
	\exp \left( \frac{\mathfrak{c}}{\sqrt{2}}Exp{\mathcal{Z}}_{T}(\mathfrak{c})-\frac{%
		\mathfrak{c}^{2}}{4}\right) &=&\left( 1+\frac{\mathfrak{c}}{\sqrt{p}}\right)
	^{-\frac{p}{2}}\int_{\Psi }\func{exp}\left( \frac{1}{2}\frac{\mathfrak{c}/%
		\sqrt{p}}{1+\mathfrak{c}/\sqrt{p}}W_{T}\left( \psi \right) \right)
	dJ\left( \psi \right)  \notag \\
	&+&  O_{p}\left( \frac{1}{p}\right) ,   \label{exp_wald}
\end{eqnarray}%
under both the null and local alternatives. For a moderate value of $%
\mathfrak{c}$, our $Exp{\mathcal{Z}}_{T}$ statistic can be approximated by an exponential statistic Exp-W with some $\mathfrak{c}^{\prime }$ by a change-of-variables. 
On the other hand,  the conventional supremum statistic for the structural break testing can also be represented by a limit case of our $Exp{\mathcal{Z}}_{T}$ since $c^{-1}\log \sum_{i=1}^n exp(cx_i) \to \max_i x_i $ as  $c\to \infty$. This is a different feature  from \cite{Andrews1994a}.

\section{A Class of Optimality Motivated HLV-Robust and Pivotal Tests}\label{sec:exp_rob_piv}

\subsection{HLV Robustification of $ \mathcal{Z}_T(\psi) $} As observed above, the HLV term $\upomega$ is a type of long run variance and hence can handled with techniques for this kind of problem. Heteroskedasticity and autocorrelation consistent/robust (HAC/HAR) inference or the estimation of the spectral density at frequency zero has a long history in time series literature, see e.g. \cite{priestley1988spectral,newey1987simple,Andrews1991a,kiefer2002heteroskedasticity,shao2015self}. Noting that the asymptotic covariance kernel of the process $ \mathcal{Z}_T(\psi) $ depends on HLV, we utilize the fixed-bandwidth kernel approach to obtain an asymptotically pivotal test.  For this purpose, $ k(\cdot) $ denotes a kernel function that satisfies Assumption \ref{ass:kernel}.



Since $\nu_{t}$ and $\Xi$ are not directly observable in practice, we replace them with the least squares-based estimates as in Section \ref{sec:model} and introduce $\ell _t = \left( np\right)^{-1/2} x_{t}'\hat{\Xi}^{-1} \hat{\nu}_{t} \text{\ensuremath{\sum_{s=1}^{t-1}x_{s}\hat{\nu}_{s}}} $ and its demeaned version,  $\bar{\ell}_{t}= \ell_t - T^{-1}\sum_{t=2}^{T} \ell_t$. A feasible kernel weighted estimate of $\upomega$ is then given by 
\begin{equation} \label{eq:Vhat}
	\hat{\upomega}=\frac{2}{T}\text{\ensuremath{\sum_{t=2}^{T}\text{\ensuremath{\sum_{s=2}^{T}k\left(\frac{t-s}{T b}\right)\bar{\ell}_{s}}}\bar{\ell}_{t}}}.
\end{equation}

\cite{Gupta2023} observed that a fixed bandwidth approach, as in \cite{Sun2014}, works well when dealing with corrections for long-run variances of nonlinear transformations of primitive variables, which is what $\upomega$ captures.  Our estimator is accordingly based on the weighting function $ K_{h}\left(r,s\right) = k\left(h\left(r-s\right)\right)$, where $ h=1/b $ and  $k\left(u\right)=\left(1-\left|u\right|\right)^{h}1\left\{ \left|u\right|<1\right\} $, using the Bartlett kernel with $h=1$. This is what is termed by \cite{Sun2014} as the sharp kernel estimator. In the theorems presented below, we characterize the joint weak limit of $ \hat{\upomega} $ and $ \mathcal{Z}_T (\psi) $, but also incorporate a bootstrap bias correction.

\subsection{Boostrap Bias Correction and Pivotalization of $ \mathcal{Z}_T(\psi) $}

Observe that the degrees of freedom $p$  provide a first-order asymptotically correct centering for $ W_T (\psi) $ but \cite{Gupta2023} show that this might induce bias in finite samples and propose a bootstrap bias correction. This is done via the \textit{null-imposed} wild bootstrap by generating 
\begin{equation} \label{eq:y^star}
	y_t^{\star} = x_t ' \hat{\kappa}_1 + \hat{u}_t \upsilon_t , \quad t=1,...,T,
\end{equation}
where $ \upsilon_t $ is an iid sequence of centered and normalized variables, e.g. Rademacher variables, to compute $ \mathcal{Z}_T^{\star} (\psi) $. The imposition of the null helps enhance the power of the test.
Repeat $ B $ times to obtain $  \bar{\mathcal{Z}}_T^{\star} (\psi) = B^{-1} \sum_j^{B} \mathcal{Z}_T^{\star, j} (\psi)  $, the bootstrap estimate of the bias. It is worth commenting that this bootstrap may not properly approximate the distribution of the test statistic even asymptotically, except the first moment that we utilize in our correction. This is due to the dependence of the asymptotic variance of the test statistic on the higher order moments of the sampling distribution. 

Then, define the bias corrected and HLV-robustified version of  $\mathcal{Z}_T (\psi)$:
\begin{equation}\label{eq:Tnb}
	\mathscr{H}_T^b (\psi) := 
	\frac{ \mathcal{Z}_T (\psi) - \bar{\mathcal{Z}}_T^{\star} (\psi) }{\sqrt{\hat{\upomega}}}.
\end{equation}
Now, let the superscript $ \star $ indicate the bootstrap analogue and $\mathscr{K}_{h}\left(r,s\right)  =K_{h}\left(r,s\right)-\int_{0}^{1}K_{h}\left(\varrho,s\right)d\varrho-\int_{0}^{1}K_{h}\left(r,\varrho\right)d\varrho +\int_{0}^{1}\int_{0}^{1}K_{h}\left(\varrho_{1},\varrho_{2}\right)d\varrho_{1}d\varrho_{2} $. Following a standard type of assumption on the kernel function we establishe pivotality of $\mathscr{H}_T^b (\psi)$.

\begin{assumption}
	\label{ass:kernel} (1) For all $x\in {\mathbb{R}}$, $k(x)=k(-x)$ and $%
	\left
	\vert k(x)\right \vert \leq 1$; $k(0)=1$; $k(x)$ is continuous at
	zero and almost everywhere on ${\mathbb{R}}$; $\int _{{\mathbb{R}}%
	}\left
	\vert k(x)\right \vert dx<\infty $. (2)
	For
	any $b\in(0,1]$ and $\rho\geq1$, $k_{b}\left(x\right)=k\left(x/b\right)$
	and $k^{\rho}\left(x\right)$ are symmetric, continuous, piecewise
	monotonic, and piecewise continuously differentiable on $\left[-1,1\right]$.
	(3) $\int_{[0,\infty )}\bar{k}%
	(x)<\infty $, where $\bar{k}(x)=\sup_{y\geq x}\left\vert k(y)\right\vert $.
\end{assumption}

\begin{theorem}\label{thm:bootstrap} Under Assumptions \ref{ass:errors}-\ref{ass:autocovandcumulant}, Assumption \ref{ass:kernel} and $ \mathcal{H}_0 $,
	\begin{equation}
		\sup_{\psi\in\Psi}\left\vert	E^{\star}W_{T}^{\star}(\psi )-p \right\vert = o_p\left(p^{1/2}\right), \label{bootthmorig}
	\end{equation}
	and
	\begin{equation} 
		\mathscr{H}^b_T(\psi)\Rightarrow \frac{\mathcal{Z}(\psi)}{\sqrt{\int_{0}^{1}\int_{0}^{1}\mathscr{K}_{h}\left(r,s\right)dG\left(r\right)dG\left(s\right)}}.\label{bootthmnew}
	\end{equation}
\end{theorem}

We introduce a class of HLV robust and pivotalized weighted
exponential statistics based on our approach that encompasses nonparametric and growing dimensional models. Define 
\begin{equation}
	Exp{\mathscr{H}}_{T}(\mathfrak{c})=\frac{\sqrt{2}}{\mathfrak{c}}\log \int_{\Psi }\func{%
		exp}\left( \frac{\mathfrak{c}}{\sqrt{2}}\mathscr{H}_{T}^b \left( \psi \right) \right) dJ\left( \psi \right)  \label{eq:expQn},
\end{equation}%
again for a positive $\mathfrak{c}$ and a bounded weight function $J(\cdot)$ such that $\int_{\Psi }dJ\left( \psi \right) =1$, and for which we set  $ \hat{\kappa}_1 = \hat{\kappa}_1 (\hat{\psi})  $ and $ \hat{\nu}_t =\hat{\nu}_t (\hat{\psi})$ in \eqref{eq:y^star}, where $ \hat{\psi} = \arg \min \hat{\sigma}(\psi) $ and  $ \hat{\sigma}(\psi) $ is the sum of squared residuals as defined in Section \ref{sec:model}. Since
\begin{equation}
	\lim_{\mathfrak{c}\rightarrow 0}Exp{\mathscr{H}}_{T}(\mathfrak{c})
	=\int_{\Psi }\mathscr{H}_{T}^b (\psi )dJ(\psi ),  \label{limit_c0}
	\quad \text{and} \quad\lim_{\mathfrak{c}\rightarrow \infty }Exp{\mathscr{H}}_{T}(\mathfrak{c})=\sup_{\psi \in \Psi }\mathscr{H}_{T}^b(\psi ),
\end{equation}%
we may extend the definition of $ Exp{\mathscr{H}}_T (\mathfrak{c} )$ for $\mathfrak{c} \in [0,+\infty]$.  

  The asymptotic distributions under both hypotheses follow from the continuous mapping theorem given the weak convergences of the stochastic process $\mathcal{Z}_T(\psi)$ on $ \Psi$ that are established in Theorems \ref{thm:null_Chow} and \ref{thm:local_power}. The main result is stated in the following theorem.
\begin{theorem}
	\label{theorem:exptest_dist} Let Assumptions \ref{ass:errors}-\ref{ass:autocovandcumulant} and Assumption \ref{ass:kernel} hold. Then, 
	\begin{equation*}
		Exp{\mathscr{H}}_{T}\left( \mathfrak{c}\right) \overset{d}{\rightarrow} \frac{\sqrt{2}}{\mathfrak{c}}%
		\log \int_{\Psi }\func{exp}\left( \frac{\mathfrak{c}\mathcal{Z}\left( \psi \right)}{\sqrt{2{ \int_{0}^{1}\int_{0}^{1}\mathscr{K}_{h}\left(r,s\right)dG\left(r\right)dG\left(s\right)}}%
		} \right) dJ\left( \psi \right) ,
	\end{equation*}%
	under ${\mathcal{H}}_{0},$ and 
	\begin{equation*}
		Exp{\mathscr{H}}_{T}\left( \mathfrak{c}\right) \overset{d}{\rightarrow} \frac{\sqrt{2}}{\mathfrak{c}}%
		\log \int_{\Psi }\func{exp}\left( \frac{\mathfrak{c}\mathcal{Z}\left( \psi \right) + \mathfrak{c}  \varrho _{\infty }\frac{\left( \psi \psi _{0}-\left( \min\left\{\psi, \psi
				_{0}\right\}\right) \right) ^{2}}{\psi (1-\psi )}
		}{\sqrt{2{{ \int_{0}^{1}\int_{0}^{1}\mathscr{K}_{h}\left(r,s\right)dG\left(r\right)dG\left(s\right)}}}%
		}
		\right) dJ\left( \psi \right)
		,
	\end{equation*}%
	under $\mathcal{H}_{\ell }$ as specified in \eqref{local_alternatives}, with $ \psi_ 0 $ denoting the true break point.
\end{theorem}
The noncentrality term is also positive for any $ \psi \in \Psi $ to make the test nontrivial.
\begin{remark}
	The test procedure can also be used if interest lies in testing for the stability of the coefficients of the $p_1$ components of $x_t=(x_{1t}',x_{2t}')'$, with $p_1<p$ and $p_1\rightarrow\infty$. Then we can test $\kappa_2=0$ in 
	\begin{equation}\label{partialbreak}
		y_t=x_{1t}'\kappa_{1}+x_{1t}'\kappa_{2}1(t/T> \psi)+x_{2t}'\gamma_2+u_t.
	\end{equation}
	The asymptotic theory in Theorem \ref{theorem:exptest_dist} can still be used, as $p_1\rightarrow\infty$ with $T\rightarrow\infty$, setting $\ell_t=\left( np\right)^{-1/2} x_{1t}'\hat{\Xi}_{11}^{-1} \hat{\nu}_{t} \text{\ensuremath{\sum_{s=1}^{t-1}x_{1s}\hat{\nu}_{s}}}$, where $\hat\Xi_{11}=T^{-1}\sum_{t}{x}_{1t} {x}_{1t} '\hat\nu_t^2$.
	
\end{remark}

\section{Sequential Limit}
\label{sec:sequential}

This section elaborates a sequential limit where first $ T \to \infty $ and then $ p \to \infty $, and the relation between the fixed $ p $ critical values and our critical values for the two extreme limiting cases of the weighted exponential tests.  We show that the HLV factor $\upomega$ is wiped out in the sequential limit. First recall that, for each $ p $, the limit has been obtained by \cite{Andrews1993}. Indeed, fixing $p$, \cite{Andrews1993}
has showed that the $W_T(\psi)={\mathcal{L}_{T}}(\psi )+o_p(1)$ uniformly in $\psi\in\Psi$, and the weak limit of the process 
\begin{equation}
	{\mathcal{L}_{T}}(\psi )= \frac{1}{ \psi \left( 1-\psi \right) T}
	\left( \sum_{t=1}^{[T\psi ]}\nu _{t}x_{t}-\psi
	\sum_{t=1}^{T}\nu _{t}x_{t}\right) ^{\prime }\Xi ^{-1}\left(
	\sum_{t=1}^{[T\psi ]}\nu _{t}x_{t}-\psi
	\sum_{t=1}^{T}\nu _{t}x_{t}\right) .
\end{equation}\label{Rdef} 
is 
\begin{equation}
	\mathcal{W}_{p}\left( \psi \right) :=\frac{\left( B_{p}\left( \psi
		\right) -\psi B_{p}\left( 1\right) \right) ^{\prime }\left( B_{p}\left(
		\psi \right) -\psi B_{p}\left( 1\right) \right) }{\psi \left( 1-\psi
		\right) },\quad \psi \in \Psi , \label{eq:mathcal W}
\end{equation}%
where $B_{p}$ stands for the $p$-dimensional standard Brownian motion and
thus $\mathcal{W}_{p}$ is the standardized tied-down Bessel process of
degree $p$. For each $\psi $, $\mathcal{W}_{p}\left( \psi \right) $ is
distributed as a Chi-square with $p$ degrees of freedom, thus with mean $%
p$ and variance $2p$. This is also a pivotal process. The following proposition derives the second limit in the sequential limit, i.e. as $ p \to \infty $, showing that it hides the HLV factor $\upomega$. 
\begin{proposition}
	\label{prop:seq}As $p\rightarrow \infty $, $
	{\left(\mathcal{W}_{p}\left( \psi \right) -p\right)}/{\sqrt{2p}}\Rightarrow {%
		\mathcal{Z}\left( \psi \right) }.
	$
\end{proposition}
A consequence of Proposition \ref{prop:seq} is that Andrews' (1993) critical values, say $%
{\alpha }$ are valid after the transformation 
\begin{equation}
	\label{andrews_transform_cv}
	\left( c_{\alpha }-p\right) \sqrt{{\upomega}/{2p}}
\end{equation}%
due to Lemma 21.2 in \cite{VanderVaart1998}, assuming that the limit
distribution function of our test statistic is continuous at the $\alpha $%
-level critical value. Let $c_{\alpha }^{\ast }$ and $c_{\alpha }^{p}$ be
the solutions of 
$
\Pr \left\{ \sup_{\psi }\mathcal{Z}\left( \psi \right) >c\right\}
=\alpha \ \ \text{and } \ \Pr \left\{ \sup_{\psi }\mathcal{W}_{p}\left( \psi
\right) >c\right\} =\alpha ,
$
respectively. Then, as 
\begin{equation}
	\Pr \left\{ \sup_{\psi }\mathcal{W}_{p}\left( \psi \right) >c_{\alpha
	}^{\ast }\sqrt{{2p}/{\upomega}}+p\right\} \rightarrow \alpha \label{eq:cstar}
\end{equation}%
as $p\rightarrow \infty ,$ by Proposition \ref{prop:seq}, we conclude from Lemma 21.2 in \cite{VanderVaart1998} that 
\begin{equation*}
	c_{\alpha }^{p}=c_{\alpha }^{\ast }\sqrt{{2p}/{\upomega}}+p+o\left(
	1\right) .
\end{equation*}
The same argument holds true for the average test. For general weighted exponential tests, the transformation for the critical values is more involved than in \eqref{eq:cstar} as given by the relationship \eqref{exp_wald} below with some approximation error. 

\section{Monte Carlo}\label{sec:mc}

In this section we examine the size and power performance of $Exp{\mathscr{H}}_{T}(\mathfrak{c})$ for the presence of a structural break at an unknown date, focusing on a special case of the exp test, viz. $Exp{\mathscr{H}}_{T} := \lim_{\mathfrak{c}\rightarrow \infty}Exp{\mathscr{H}}_T(\mathfrak{c}) $, see (\ref{limit_c0}). Recall that it is equivalent to  $\sup_{\psi\in\Psi}\mathscr{H}_{T}^{b}\left(\psi\right)$. We compare it with $Exp{\mathcal{Z}}_{T} := \lim_{\mathfrak{c}\rightarrow \infty}Exp{\mathcal{Z}}_T(\mathfrak{c}) $
and $supW_T :=\sup_{\psi\in\Psi}W_{T}\left(\psi\right)$, where the critical
values for the last test come from \cite{Andrews1993}. 
The experiments for size cover $T=250,500,$ while for power we present results for $T=400$. Finally, we set $\Psi=[0.15,0.85]$.

We consider two examples: multiple regression and infinite order autoregression or the ARMA model, for which we generate the error  from a bounded ARCH process 
\begin{equation}\label{eq:ARCH}
	\nu_{t}  =\sigma_{t}\bar\eta_{t},\;\;\;\;\;\;\;\;\;\;\;\;\;\sigma_{t}^{2}  =\left(1-\alpha\right)+\alpha f ( \nu_{t-1} ) ,
\end{equation}
where $ f(x)  = x^2 1\{|x|\leq b\} +b^2 1\{|x|> b\} $, $\bar\eta_{t}=\left(\eta_{t}-E\eta_{t}\right)/\sqrt{var\left(\eta_{t}\right)}$,
and $\left\{ \eta_{t}\right\} $ is an iid sequence from the  normal mixture distributions of type 1 and 2 in \cite{Marron1992}. 
We fix $b=2.5$. 
Their mixture type 1 is the standard normal. For a standard normal vector $\left(Z_{1},...,Z_{k}\right)$ and multinomial vector $\left(d_{1},...,d_{k}\right)$ with probability $\left(1/5,1/5,3/5\right)$, the mixture type 2 skewed unimodal variate is $\eta_{t}=Z_{1}d_{1}+\left(2Z_{2}/3+1/2\right)d_{2}+d_{3}\left(5Z_{3}/9+13/12\right)$. We also experimented with their mixture type 3 but do not report this here because  the results are similar. 

More specifically, for the multiple regression,  the regressors $ x_t $ consist of independent AR(1) processes with coefficient $ \alpha_x $ and the ARCH innovation as in \eqref{eq:ARCH} and their lags of order up to 3. That is, we consider the distributed lag (DL) model with growing numbers of variables. The first five elements of the coefficients are set as $d_0\left(5^{-1/2},...,5^{-1/2}\right)p^{1/4}T^{-1/2}$ and the others as zeros. When there is a break, all the values become zero after the break so that the value $ d_0 $ controls the magnitude of the change. 
We vary $\alpha \in \{0.3,0.55\}$, and $p \in \{5,9,13\},$ to examine the effect of the dimension on our tests. 
For the infinite order AR regression,  we generate the sample from the MA(1) model $ y_{t}=\nu_{t}+\theta1\left\{ t\leq\mu\right\} \nu_{t-1} $, $\mu=T\psi$, and estimate the AR($ p $) model with $p=9$ for $ T = 250 $ and $ p = 18 $ for $ T = 500 $. Note that the coefficient $\theta$ represents the size of the jump as well as the MA coefficient.

\begin{table}[tbhp]
\centering
\caption{Size Comparisons: Multiple Regression }

\begin{tabular}{l|cccc|cccc}
\toprule
\multicolumn{9}{c}{$n=250$} \\ 
\midrule
      & \multicolumn{4}{c|}{$p = 5$} & \multicolumn{4}{c}{$p=9$} \\($\eta$, $\alpha$)
      & (1,0.3) & (1,0.55) & (2,0.3) & (2,0.55)
      & (1,0.3) & (1,0.55) & (2,0.3) & (2,0.55)
      \\
\midrule
ExpH  & 0.130 & 0.126 & 0.114 & 0.103 & 0.128 & 0.116 & 0.112 & 0.110 \\
ExpZ  & 0.108 & 0.109 & 0.105 & 0.111 & 0.237 & 0.233 & 0.248 & 0.256 \\
supW  & 0.169 & 0.166 & 0.168 & 0.185 & 0.430 & 0.422 & 0.415 & 0.430 \\
\midrule 
\multicolumn{9}{c}{$n=500$} \\
\midrule
    & \multicolumn{4}{c|}{$p=9$} & \multicolumn{4}{c}{$p=13$}\\($\eta$, $\alpha$)
      & (1,0.3) & (1,0.55) & (2,0.3) & (2,0.55)
      & (1,0.3) & (1,0.55) & (2,0.3) & (2,0.55)\\
\midrule
ExpH  & 0.050 & 0.051 & 0.052 & 0.051 & 0.041 & 0.039 & 0.036 & 0.025 \\
ExpZ  & 0.057 & 0.065 & 0.070 & 0.072 & 0.117 & 0.123 & 0.117 & 0.119 \\
supW  & 0.156 & 0.159 & 0.172 & 0.179 & 0.315 & 0.307 & 0.303 & 0.321 \\
\bottomrule
\end{tabular}

Note: $\Psi = [0.15, 0.85]$, Nominal level is 5\%.
\label{tab:TSsupsizemult}
\end{table}

The vertical partitions in Table \ref{tab:TSsupsizemult} represent different values of $p$, increasing from left to right. The traditional sup test $supW_T$  over-rejects under both $T=250$ and $500$, with performance getting poorer as $p$ increases. The correction for $p$ immediately improves matters, with $Exp{\mathcal{Z}}_T $ still over-rejecting but to a lesser degree. Our recommended statistic $Exp{\mathscr{H}}_T$ performs best overall, and achieves quite acceptable size performance when the sample size is bigger with $T=500$ even when $p$ is moderate. Much the same lesson is learnt from Table \ref{tab:TSsupsizeARMA}, that is, failure to correct for the HLV factor $\upomega$ or the finite sample bias leads to size distorted tests. We note that the $supW_T$ and $Exp{\mathcal{Z}}_T$ tests are severely undersized while our recommended statistic $Exp{\mathscr{H}}_T$ maintains reasonable rejection rates that are much closer to the nominal $5$ \% level.    Here, we only report the case of $T=500$ to save space since the rejection rates for $supW_T$ and $Exp{\mathcal{Z}}_T$ are more than 0.5  when the $T=250$.

\begin{table}[tbhp]
\centering
\caption{Size Comparison: ARMA}
\begin{tabular}{l|cccc|cccc}
\toprule
      & \multicolumn{4}{c|}{$\theta = -0.4$} & \multicolumn{4}{c}{$\theta = 0.8$}  \\ 
      ($\eta$, $\alpha$)
      & (1,0.3) & (1,0.55) & (2,0.3) & (2,0.55)
      & (1,0.3) & (1,0.55) & (2,0.3) & (2,0.55)
      \\
\midrule
ExpH  & 0.067 & 0.053 & 0.046 & 0.045 & 0.070 & 0.044 & 0.046 & 0.040 \\
ExpZ  & 0.002 & 0.003 & 0.007 & 0.008 & 0.002 & 0.003 & 0.003 & 0.004 \\
supW  & 0.008 & 0.008 & 0.015 & 0.018 & 0.014 & 0.008 & 0.014 & 0.018 \\
\bottomrule
\end{tabular}

Note: $n = 500$, $\Psi = [0.15, 0.85]$, Nominal level is 5\%.
\label{tab:TSsupsizeARMA} 
\end{table}

Finally, each panel of Table \ref{tab:TSsuppowerbotth} displays power for both regressions, showing a vertical partition that corresponds to increasing break magnitude from left to right. 
Specifically, the partitions correspond to $d_0=1,5,10$ for the multiple regression and $\theta = 0.4,0.6,0.8$ for ARMA. These parameters control the size of the change of the coefficients after the break. Within each partition the parameters  vary lexicographically as $(\eta, \alpha)$. We observe that the power of $Exp{\mathscr{H}}_T$ and the other statistics grow as the magnitude of the structural break increases. Across different tests, there is no dominant test in terms of power. It may reflect large size distortions of the $supW_T$ and $Exp{\mathcal{Z}}_T$ tests, but our recommended statistic $Exp{\mathscr{H}}_T$ does not unduly sacrifice power.

\begin{table}[htbp]
\centering
\caption{Power Comparison}
\resizebox{\textwidth}{!}{%
\begin{tabular}{l|cccc|cccc|cccc}
\toprule
\multicolumn{13}{c}{Multiple Regression} \\
\midrule
      & \multicolumn{4}{c|}{$d_0 = 1$} & \multicolumn{4}{c|}{$d_0 = 5$} & \multicolumn{4}{c}{$d_0 = 10$} \\
      ($\eta$, $\alpha$)
      & (1,.3) & (1,.55) & (2,.3) & (2,.55)
      & (1,.3) & (1,.55) & (2,.3) & (2,.55)
      & (1,.3) & (1,.55) & (2,.3) & (2,.55)\\
\midrule
ExpH  & 0.097 & 0.065 & 0.096 & 0.066 & 0.635 & 0.597 & 0.638 & 0.596 & 0.994 & 0.992 & 0.994 & 0.993 \\
ExpZ  & 0.116 & 0.123 & 0.116 & 0.123 & 0.813 & 0.814 & 0.813 & 0.814 & 1.000 & 1.000 & 1.000 & 1.000 \\
supW  & 0.260 & 0.256 & 0.260 & 0.256 & 0.923 & 0.928 & 0.923 & 0.928 & 1.000 & 1.000 & 1.000 & 1.000 \\
\midrule
\multicolumn{13}{c}{ARMA} \\
\midrule
       & \multicolumn{4}{c|}{$\theta=0.4$} & \multicolumn{4}{c|}{$\theta=0.6$} & \multicolumn{4}{c}{$\theta=0.8$} \\
     ($\eta$, $\alpha$)
      & (1,.3) & (1,.55) & (2,.3) & (2,.55)
      & (1,.3) & (1,.55) & (2,.3) & (2,.55)
      & (1,.3) & (1,.55) & (2,.3) & (2,.55) \\
\midrule
ExpH  & 0.195 & 0.191 & 0.150 & 0.150 & 0.518 & 0.453 & 0.360 & 0.337 & 0.822 & 0.772 & 0.657 & 0.613 \\
ExpZ  & 0.045 & 0.037 & 0.024 & 0.032 & 0.258 & 0.241 & 0.147 & 0.178 & 0.653 & 0.679 & 0.507 & 0.496 \\
supW  & 0.132 & 0.115 & 0.110 & 0.101 & 0.498 & 0.464 & 0.341 & 0.344 & 0.844 & 0.852 & 0.730 & 0.710 \\
\bottomrule
\end{tabular}
}
Note: $\Psi = [0.15, 0.85]$. Nominal level is 5\%. $n=400$ 
\label{tab:TSsuppowerbotth}
\end{table}

\section{Empirical Application}\label{sec:application} 

We apply our test to examine the stability of the dynamics of economic output variables and the oil price–output nexus, originally examined in \citet{Hamilton2003}, using an autoregressive distributed lag model (ADL) of order $(p,p)$. We employ quarterly U.S. macroeconomic data, where the growth rate of chain-weighted real GDP serves as the measure of real activity. Oil price dynamics are captured through the nominal crude oil producer price index, not seasonally adjusted.	

Following Hamilton's formulation, we consider three oil price measures: (i) the quarterly percentage change in oil prices, denoted $o_t$; (ii) a one-sided transformation capturing only positive changes, $o_t^+ = o_t \cdot \mathbf{1}\{o_t > 0\}$; and (iii) the net oil price increase, $o_t^T$, which equals the excess of log oil prices in period $t$ over the peak value observed in the prior 12 months, or zero otherwise. The sample spans from 1949:Q1 to 2019:Q4 and is obtained from the FRED database of the St. Louis Federal Reserve. 

As a sensitivity measure, we extend the analysis to monthly data using the industrial production (IP) index as an alternative measure of economic activity. Given the monthly frequency, we employ the ADL$(12,12)$ and ADL$(18,18)$ models, which correspond to one-year and 18-month lag structures, respectively. The shift to monthly data substantially increases the number of restrictions involved in testing for structural change. The number of restrictions vary from 13 in AR(12) and 25 in ADL(12,12), to 19 in AR(18) and 37 in ADL(18,18), while the number is at most 13 in case of the quarterly series. This growth in dimensionality has important implications for inference. Notably, critical values for the standard $supW_T$ test are not tabulated for degrees of freedom exceeding 20 in \citet{Andrews1993} or \citet{Hansen1997a}, further motivating the use of our alternative procedure.

We first test for structural instability in the univariate series using AR$(p)$ models. Subsequently, we estimate ADL$(p,p)$ regressions of GDP growth or industrial production index on oil prices using the three oil price specifications, and examine potential structural breaks in these relationships. In line with \citet{Hamilton2003}, we do not fix a specific break date but adopt $Exp{\mathscr{H}}_T $ statistics evaluated over a trimmed interval $\Psi = [0.15, 0.85]$ as in Section \ref{sec:mc}. And we compare it with  $supW_T$, where the latter's critical values are computed using the R function provided in \citet{Hansen1997a} and they are not accounted for the issues arising from large $p$.

Table~\ref{table:examplesubsamp} reports the p-values of both tests for various specifications. The results reveal divergent inference: while the standard $supW_T$ test supports the existence of structural breaks in both AR and ADL regressions, the robust $Exp{\mathscr{H}}_T$ test yields mixed evidence, suggesting that conventional methods may overstate the presence of instability in this context. 
The results for the full sample mirror the patterns observed in the Monte Carlo simulation. However, the discrepancies between the robust and conventional test statistics are magnified in the settings of larger $p$s, reinforcing concerns that standard procedures may yield misleading inferences in high-dimensional applications. 

We also conduct subsample analyses to assess the temporal robustness of our findings to account for potential structural shifts in macroeconomic dynamics, such as changes in monetary policy regimes or energy market structures. Specifically, we consider two subsamples: Sub 1 starts after the 1980 oil shock while SS2 ends before the 2007 oil shock. As for the lag order $ p $, we try both $ p=4 $ and $ p=6 $ for GDP dynamice and  $p=12$ and $p=18$ for IP dynamics as before. For each subsample, we re-estimate the AR and ADL models and apply both the standard and robust structural break tests.

\begin{table}[thb]
	\centering
    \caption{Tests for stability of  AR$(p)$ and ADL$(p,p)$}
    \resizebox{\textwidth}{!}{

    \begin{tabular}{c|cc|cc|cc|cc|cc|cc|cc|cc}
    \toprule
    & \multicolumn{8}{c|}{\textbf{GDP}} & \multicolumn{8}{c}{\textbf{IP}} \\ 
    \midrule
    &\multicolumn{2}{c|}{(a) AR($p$)} &\multicolumn{6}{c|}{(b) ADL($p,p$)} &\multicolumn{2}{c|}{(a) AR($p$)} &\multicolumn{6}{c}{(b) ADL($p,p$)} \\
    &&&\multicolumn{2}{c}{$o_t$}&\multicolumn{2}{c}{$o_t^+$}&\multicolumn{2}{c|}{$o_t^T$} & & &\multicolumn{2}{c}{$o_t$}&\multicolumn{2}{c}{$o_t^+$}&\multicolumn{2}{c}{$o_t^T$}\\
	\midrule   
    lags $p$ & 4 & 6 &  4 & 6 & 4 & 6 & 4 &6  & 12 & 18 & 12 & 18& 12 & 18& 12 & 18 \\
    \midrule \midrule
     \multicolumn{17}{c}{Full Sample} \\ \midrule
			
	ExpH & 40.4 & 22.4&7.2 &2.9&0.08&1.7&0&0   & 35.6 & 53.4&4.3 &6.6&1&1.4&1.4&0.9\\ 
	  supW & 7.2 & 4.1&0&0&2&0&0&0 & 0 & 0.5&0&0&0&0&0&0\\
	\midrule
     \multicolumn{17}{c}{Sub Sample 1} \\ \midrule
 	ExpH  & 8.7 & 0.9&3.8 &3.1&0.2&15.8&3.9&11.4   & 55.7 & 3.8&1.9 &3.3&4.4&11&73.9&54.9\\
    supW & 0.7 & 0&0&0&0&0&0&0 & 23.4 & 0.2&0&0&0&0&0&0\\
	\midrule
     \multicolumn{17}{c}{Sub Sample 2} \\ \midrule
	  ExpH  & 46.6 & 26.1&4.3 &3.3&9.1&2&0&0   & 26.4 & 45.6&3.6 &5.6&0.6&0.4&1.6&0.3\\
	supW & 45.3 & 7.1&0&0&0.1&0&0&0  & 0 & 0&0&0&0&0&0&0\\
    \bottomrule 
    \end{tabular}}
    	        
    Note: $ 100 \times $p-values of stability tests for full sample and subsamples. Sub 1: 1981:I-2019:IV, Sub 2: 1950:I-2007:II; $\psi\in[0.15,0.85]$. (a) Tests for stability of  AR$(p)$ fits. (b) Tests for stability of ADL$(p,p)$ regressions on $o_t$, $o_t^+$ or $o_t^T$. 
    \label{table:examplesubsamp}
\end{table}

The results are mostly similar to those observed in the full sample. The standard Wald test supports the presence of instability in the oil-output relations most of the time, while our test tells mixed story depending on the specifications. Also, the instabilities in the dynamics of GDP growth rate and industrial production have become substantially weakened. 
Contrary to the expectation that structural instability might be more pronounced in the earlier period, the p-values across subsamples do not exhibit a systematic pattern. In several cases, the robust test detects weaker evidence of breaks in SS1 relative to SS2, suggesting that the magnitude and nature of instability may vary by specification and measure of oil prices. These findings underscore the importance of flexible testing procedures that do not presuppose the direction or timing of instability. They also caution against broad characterizations of macroeconomic stability across periods without formal statistical validation.

\section{Conclusion}

We have developed a class of optimal tests for structural breaks occurring at unknown dates in infinite and growing-dimensional linear time series models. Extending the exponential-average framework of \cite{Andrews1994} to the large-dimensional setting, we established average power optimality under Gaussian designs and derived a functional central limit theorem that remains valid when both the sample size and the number of restrictions diverge. The limiting process features a nonlinear high-order long-run variance term that captures interactions between serial dependence and dimensional growth, a phenomenon absent in fixed-dimensional models.

We proposed a random-scaling correction and a bootstrap bias adjustment that jointly yield a pivotal, HLV-robust test. The resulting statistic maintains correct asymptotic size and performs well in finite samples. Monte Carlo experiments demonstrate substantial gains in size control relative to conventional supremum-type and exponential tests, particularly when the number of regressors is large. The framework integrates optimal testing principles with growing-dimension asymptotics and provides a basis for more robust and interpretable inference than standard tests, highlighting its potential value in high-dimensional data analysis where conventional asymptotics are unreliable.

\begin{appendix}




\section{Proof of Theorem \ref{theorem:optimality}}
\begin{proof}[Proof of Theorem \protect\ref{theorem:optimality}]
	We begin with establishing an asymptotic equivalence under both the null and
	alternative hypotheses. Since the Taylor series expansion yields $
		\log \left( 1+\frac{\mathfrak{c}}{m}\right) ^{-\frac{m}{2}} = -\frac{m}{2}%
		\log \left( 1+\frac{\mathfrak{c}}{m}\right) 
		= -\frac{\mathfrak{c}}{2}+\frac{\mathfrak{c}^{2}}{4m}-\frac{\mathfrak{c}^{3}%
		}{6m^{2}}+O\left( \frac{1}{m^{3}}\right) $
	for a large $m$, letting $m=\sqrt{p}$ and applying an exponential
	transformation $\left( \func{exp}\left( \cdot \right) \right) ^{\sqrt{p}}$
	for both sides yield 
	\begin{equation*}
		\frac{\left( 1+\frac{\mathfrak{c}}{\sqrt{p}}\right) ^{-\frac{p}{2}}}{\left( -%
			\frac{\mathfrak{c}\sqrt{p}}{2}\right) }=\func{exp}\left( \frac{\mathfrak{c}%
			^{2}}{4}-\frac{\mathfrak{c}^{3}}{6\sqrt{p}}+O\left( \frac{1}{p}\right)
		\right) =\func{exp}\left( \frac{\mathfrak{c}^{2}}{4}\right) O\left( e^{-%
			\frac{\mathfrak{c}^{3}}{6\sqrt{p}}+O\left( \frac{1}{p}\right) }\right)
	\end{equation*}%
	as $p\rightarrow \infty .$ On the other hand, 
	\begin{eqnarray*}
		\frac{1}{2}\frac{\mathfrak{c}/\sqrt{p}}{1+\mathfrak{c}/\sqrt{p}}W_{T}\left(
		\psi \right) -\frac{\mathfrak{c}\sqrt{p}}{2} &=&\frac{1}{\sqrt{2}}\frac{%
			\mathfrak{c}}{1+\mathfrak{c}/\sqrt{p}}\left( \frac{W_{T}\left( \psi
			\right) -p}{\sqrt{2p}}\right) -\frac{\mathfrak{c}^{2}}{2\left( 1+\mathfrak{c}%
			/\sqrt{p}\right) } \\
		&=&\left( \frac{\mathfrak{c}}{\sqrt{2}}\mathcal{Z}_{T}\left( \psi \right) -%
		\frac{\mathfrak{c}^{2}}{2}\right) \frac{1}{1+\mathfrak{c}/\sqrt{p}}.
	\end{eqnarray*}%
	Thus, putting these together leads to 
	\begin{eqnarray}
		ExpW_{T}\left( \mathfrak{c}\right) &:=&\left( 1+\frac{\mathfrak{c}}{\sqrt{p}%
		}\right) ^{-\frac{p}{2}}\int \exp \left( \frac{1}{2}\frac{\mathfrak{c}/\sqrt{%
				p}}{1+\mathfrak{c}/\sqrt{p}}W_{T}\left( \psi \right) \right) dJ\left(
		\psi \right)  \notag  \\
		&=&\int \func{exp}\left( \left( \frac{\mathfrak{c}}{\sqrt{2}}\mathcal{Z}%
		_{T}\left( \psi \right) -\frac{\mathfrak{c}^{2}}{4}\right)  \frac{%
			1}{1+\mathfrak{c}/\sqrt{p}} \right) dJ\left(
		\psi \right) \exp \left( -\frac{\mathfrak{c}^{3}}{6\sqrt{p}}+O\left( \frac{%
			1}{p}\right) \right) \label{eq:expw} \\
		&=&\int \func{exp}\left( \frac{\mathfrak{c}}{%
			\sqrt{2}}\mathcal{Z}_T\left( \psi \right) -\frac{\mathfrak{c}^{2}}{4}\right)
		dJ\left( \psi \right) + o_p(1)  \notag
	\end{eqnarray}%
	by the CMT, the weak convergences of $\mathcal{Z}_{T}\left( \psi \right) $
	in Theorem \ref{thm:null_Chow} under $H_{0}$ and in Theorem \ref%
	{thm:local_power} under $H_{\ell }$. This shows that the weak limit of $%
	ExpW_{T}\left( \mathfrak{c}\right) $ coincides with that of $\exp \left( 
	\frac{\mathfrak{c}}{\sqrt{2}}Exp{\mathcal{Z}}_{T}(\mathfrak{c})-\frac{\mathfrak{c}^{2}}{4%
	}\right) =\int \func{exp}\left( \frac{\mathfrak{c}}{\sqrt{2}}\mathcal{Z}%
	_{T}\left( \psi \right) -\frac{\mathfrak{c}^{2}}{4}\right) dJ\left( \psi
	\right) $.
	
	Next, we establish the asymptotic equivalence between $ExpW_{T}$ and $LR_{T}$
	. Hereafter, we drop the argument $\mathfrak{c}$ as it does not cause any
	confusion. Let $E$ denote the expectation under the null and $X=\left(
	x_{1},...,x_{T}\right) $. When $X$ is strictly exogeneous, the Fubini
	theorem and the law of iterated expectations yield that 
	\begin{eqnarray*}
		&&E\int \int_{\left\vert \boldsymbol{b}_{1}\right\vert \geq \log p}\func{exp}%
		\left( \frac{-1}{2\sigma _{0}^{2}}\sum_{t=1}^{T}\left( y_{t}-\left( T%
		\sqrt{p}\right) ^{-1/2}\psi _{t}x_{t}^{\prime }\boldsymbol{b}\right)
		^{2}-y_{t}^{2}\right)  dQ\left( \boldsymbol{b}\right) dJ\left( \psi
		\right) \\
		&=&E\int \int_{\left\vert \boldsymbol{b}_{1}\right\vert \geq \log p}E\left[ 
		\func{exp}\left( \frac{-1}{2\sigma _{0}^{2}}\sum_{t=1}^{T} \left(
		y_{t}-\left( T\sqrt{p}\right) ^{-1/2}\psi _{t}x_{t}^{\prime }\boldsymbol{b}%
		\right) ^{2}-y_{t}^{2}\right)  |X\right] dQ\left( \boldsymbol{b}%
		\right) dJ\left( \psi \right) \\
		&=&\int \int_{\left\vert \boldsymbol{b}_{1}\right\vert \geq \log
			p}dQ_{\psi }\left( \boldsymbol{b}\right) dJ\left( \psi \right) \rightarrow 0,
	\end{eqnarray*}%
	since the expected value of a likelihood ratio is 1 under the null.
	Thus, 
	\begin{eqnarray*}
		LR_{T}=\underset{:=LR_{1}}{\underbrace{\int \int_{\left\vert \boldsymbol{b}%
					_{1}\right\vert <\log p}\func{exp}\left( \frac{-1}{2\sigma _{0}^{2}}%
				\sum_{t=1}^{T} \left( y_{t}-\left( T\sqrt{p}\right) ^{-1/2}\psi
				_{t}x_{t}^{\prime }\boldsymbol{b}\right) ^{2}-y_{t}^{2} \right)
				dQ\left( \boldsymbol{b}\right) dJ\left( \psi \right) }}
		+ o_{p}\left(
		1\right) .
	\end{eqnarray*}    
    Let $m_{t}=\left( T\sqrt{p}\right) ^{-1/2}\psi _{t}x_{t}^{\prime }%
	\boldsymbol{b}_{1}$ and $r_{t}=\left( T\sqrt{p}\right) ^{-1/2}\psi
	_{t}w_{t}^{\prime }\boldsymbol{b}_{2}$. Then, 
	\begin{equation*}
		\left( y_{t}-\left( T\sqrt{p}\right) ^{-1/2}\psi _{t}x_{t}^{\prime }%
		\boldsymbol{b}\right) ^{2}-y_{t}^{2}=\left( \left( y_{t}-m_{t}\right)
		^{2}-y_{t}^{2}\right) +r_{t}^{2}-2r_{t}y_{t}+2r_{t}m_{t}.
	\end{equation*}%
	For each $\psi ,$ it follows from Fubini's theorem that for strictly
	exogeneous $X$ 
	\begin{eqnarray*}
		&&E\left[ \int \func{exp}\left( -\frac{1}{2\sigma _{0}^{2}}%
		\sum_{t=1}^{T}\left( r_{t}^{2}-2r_{t}y_{t}\right) \right) dQ_{2}\left( 
		\boldsymbol{b}_{2}|\boldsymbol{b}_{1}\right) |X\right] \\
		&=&\int \cdots \int \left( 2\pi \sigma _{0}^{2}\right) ^{-T/2}\func{exp}%
		\left( -\frac{1}{2\sigma _{0}^{2}}\sum_{t=1}^{T}\left( r_{t}-y_{t}\right)
		^{2}\right) dy_{1}\cdots dy_{T}dQ_{2}\left( \boldsymbol{b}_{2}|\boldsymbol{b}%
		_{1}\right) =1,
	\end{eqnarray*}%
	since the integrand is a density of multivariate normal.
	
	Also, note that $y_{t}=\sigma _{0}\nu _{t}$ being normal conditional
	on $x_{t}$ and $r_{t}$ is a function of $x_{t}$ imply that 
	\begin{equation}
		E\left[ \left( -\frac{1}{2\sigma _{0}^{2}}\left(
		r_{t}^{2}-2r_{t}y_{t}\right) \right) |x_{t}\right] =\int \frac{1}{\sqrt{2\pi
				\sigma _{0}^{2}}}\exp \left( -\frac{1}{2\sigma _{0}^{2}}\left(
		r_{t}-y_{t}\right) ^{2}\right) dy_{t}=1  \label{eq:AR_1}
	\end{equation}%
	as the integrand is a density of normal distribution. Then, it follows from
	Fubini's theorem and repeatedly applying (\ref{eq:AR_1}) by the law
	of iterated expectations that 
	\begin{eqnarray*}
		&&E\left[ \int \func{exp}\left( -\frac{1}{2\sigma _{0}^{2}}%
		\sum_{t=1}^{T}\left( r_{t}^{2}-2r_{t}y_{t}\right) \right) dQ_{2}\left( 
		\boldsymbol{b}_{2}|\boldsymbol{b}_{1}\right) \right] \\
		&=&\int E\func{exp}\left( -\frac{1}{2\sigma _{0}^{2}}\sum_{t=1}^{T}\left(
		r_{t}^{2}-2r_{t}y_{t}\right) \right) dQ_{2}\left( \boldsymbol{b}_{2}|%
		\boldsymbol{b}_{1}\right) \\
		&=&\int E\left[ \func{exp}\left( -\frac{1}{2\sigma _{0}^{2}}%
		\sum_{t=1}^{T-1}\left( r_{t}^{2}-2r_{t}y_{t}\right) \right) E\left[ \left( -%
		\frac{1}{2\sigma _{0}^{2}}\left( r_{T}^{2}-2r_{T}y_{T}\right) \right) |x_{T}%
		\right] \right] dQ_{2}\left( \boldsymbol{b}_{2}|\boldsymbol{b}_{1}\right),
	\end{eqnarray*}
	which is 1. 
	Proceeding similarly, and in view of Jensen's inequality 
	\begin{eqnarray*}
		&&E\left( \int \func{exp}\left( -\frac{1}{2\sigma _{0}^{2}}%
		\sum_{t=1}^{T}\left( r_{t}^{2}-2r_{t}y_{t}\right) \right) dQ_{2}\left( 
		\boldsymbol{b}_{2}|\boldsymbol{b}_{1}\right) \right) ^{2} \\
		&\leq &E\int \func{exp}\left( -\frac{1}{\sigma _{0}^{2}}\sum_{t=1}^{T}\left(
		r_{t}^{2}-2r_{t}y_{t}\right) \right) dQ_{2}\left( \boldsymbol{b}_{2}|%
		\boldsymbol{b}_{1}\right) \\
		&=&\int E\left[ E\left[ \int \cdots \int \left( 2\pi \sigma _{0}^{2}\right)
		^{-T/2}\func{exp}\left( -\sum_{t=1}^{T}\frac{\left( 2r_{t}-y_{t}\right) ^{2}%
		}{2\sigma _{0}^{2}}\right) dy_{1}\ldots dy_{T}|X\right] \right.\\
		&\times&\left.\func{exp}\left(
		\sum_{t=1}^{T}\frac{r_{t}^{2}}{\mathfrak{\ \sigma }_{0}^{2}}\right) \right]
		dQ_{2}\left( \boldsymbol{b}_{2}|\boldsymbol{\ b}_{1}\right) \\
		&=&\int E\func{exp}\left( \sum_{t=1}^{T}\frac{r_{t}^{2}}{\sigma _{0}^{2}}%
		\right) dQ_{2}\left( \boldsymbol{b}_{2}|\boldsymbol{b}_{1}\right)
		\rightarrow 1,
	\end{eqnarray*}%
	uniformly in $\psi $ and $\boldsymbol{b}_{1}$ by Assumption \ref{ass:opt1}%
	. For the convergence, note that $\left\vert \boldsymbol{b}_{2}\right\vert
	_{2}^{2}$ is uniformly bounded and $\left\vert \boldsymbol{b}_{2}\right\vert
	_{2}^{2}\rightarrow 0$ as $T\rightarrow \infty $, and that 
	\begin{eqnarray*}
		E\func{exp}\left( \frac{1}{\sigma _{0}^{2}}\sum_{t=1}^{T}r_{t}^{2}\right)
		&\leq &\sup_{\left\vert a\right\vert _{2}=1}E\exp \left( \frac{1}{T\sqrt{p}%
			\sigma _{0}^{2}}\sum_{t=1}^{T}\left\vert w_{t}^{\prime }a\right\vert
		^{2}\left\vert \boldsymbol{b}_{2}\right\vert _{2}^{2}\right) \\
		&\leq &\left( \max_{t}\sup_{\left\vert a\right\vert _{2}=1}E\exp \left( 
		\frac{1}{\sqrt{p}\sigma _{0}^{2}}\left\vert w_{t}^{\prime }a\right\vert
		^{2}\left\vert \boldsymbol{b}_{2}\right\vert _{2}^{2}\right) \right) \\
		&\leq &\left( \max_{t}\sup_{\left\vert a\right\vert _{2}=1}E\exp \left(
		\left\vert w_{t}^{\prime }a\right\vert ^{2}\right) \right) ^{\left\vert 
			\boldsymbol{b}_{2}\right\vert _{2}^{2}/\left( \sigma _{0}^{2}\sqrt{p}\right)
		},
	\end{eqnarray*}%
	where the second inequality is due to Assumption \ref{ass:opt1} and the last
	is by Jensen's inequality as $\left\vert \boldsymbol{b}_{2}\right\vert
	_{2}^{2}/\sqrt{p}$ is less than 1, and then apply the dominated convergence
	theorem. As we have shown that both the first and second moments converge to 
	$1$ uniformly in $\psi $ and $\boldsymbol{b}_{1}$, we get 
	\begin{equation}
		\int \func{exp}\left( -\frac{1}{2\sigma _{0}^{2}}\sum_{t=1}^{T}\left(
		r_{t}^{2}-2r_{t}y_{t}\right) \right) dQ_{2}\left( \boldsymbol{b}_{2}|
		\boldsymbol{b}_{1}\right) \overset{p}{\longrightarrow }1,  \label{eq:rtyt}
	\end{equation}
	uniformly in $\psi $ and $\boldsymbol{b}_{1}$. Also, 
	\begin{equation*}
		\sup_{\psi ,\boldsymbol{b}_{2}}\sup_{\left\vert \boldsymbol{b}_{1}\right\vert <\log p}\func{exp}\left( \frac{1}{\sigma _{0}^{2}}
		\sum_{t=1}^{T}\left\vert r_{t}m_{t}\right\vert \right) =\func{exp}\left(
		\sup_{\psi ,\boldsymbol{b}_{2}}\sup_{\left\vert \boldsymbol{b}_{1}\right\vert <\log p}\left( \frac{1}{\sigma _{0}^{2}}\sum_{t=1}^{T}\left\vert r_{t}m_{t}\right\vert \right) \right) \overset{p}{\longrightarrow }1,
	\end{equation*}
	since 
	\begin{eqnarray*}
		\sup_{\psi ,\boldsymbol{b}_{2}}\sup_{\left\vert \boldsymbol{b}
			_{1}\right\vert <\log p}\sum_{t=1}^{T}\left\vert r_{t}m_{t}\right\vert \leq
		\sup_{\left\vert \boldsymbol{b}_{1}\right\vert <\log p}\left\vert 
		\boldsymbol{b}_{1}\right\vert \sup_{\boldsymbol{b}_{2}}\left\vert
		w_{t}^{\prime }\boldsymbol{b}_{2}\right\vert \frac{1}{T}\sum_{t=1}^{T}\left\vert \frac{1}{p}x_{t}^{\prime }x_{t}\right\vert ^{1/2} = o_{p}\left( p^{-1/2}\log p\right) ,
	\end{eqnarray*}
	where the equality follows from the fact that $E\left\vert \frac{1}{p}
	x_{t}^{\prime }x_{t}\right\vert \ $and $E\sup_{\left\vert a\right\vert
		_{2}=1}\left\vert w_{t}^{\prime }a\right\vert $ are bounded by Assumption \ref{ass:opt1} and thus 
	\begin{equation*}
		\sup_{\boldsymbol{b}_{2}}E\left\vert w_{t}^{\prime }\boldsymbol{b}_{2}\right\vert \leq E\sup_{\left\vert a\right\vert _{2}=1}\left\vert
		w_{t}^{\prime }a\right\vert \sup_{\boldsymbol{b}_{2}}\left\vert \boldsymbol{b}_{2}\right\vert _{2}=O\left( T^{-1/2}\right) .
	\end{equation*}
	This in turn implies that 
	\begin{equation}
		\sup_{\psi ,\boldsymbol{b}_{2}}\sup_{\left\vert \boldsymbol{b}_{1}\right\vert <\log p}\left\vert \func{exp}\left( \frac{1}{\mathfrak{\
				\sigma }_{0}^{2}}\sum_{t=1}^{T}r_{t}m_{t}\right) -1\right\vert =o_{p}\left(
		1\right) ,  \label{eq:rtmt}
	\end{equation}
	since $\max_{x}\left\vert 1-e^{f_{T}\left( x\right) }\right\vert \leq
	\max_{x}e^{\left\vert f_{T}\left( x\right) \right\vert }-1.$ Then, by (\ref%
	{eq:rtmt}) 
	\begin{eqnarray*}
		LR_{1} &=&\int \int_{\left\vert \boldsymbol{b}_{1}\right\vert <\log p}\int 
		\func{exp}\left( -\frac{1}{2\sigma _{0}^{2}}\sum_{t=1}^{T}\left(
		r_{t}^{2}-2r_{t}y_{t}\right) \right) dQ_{2}\left( \boldsymbol{b}_{2}|%
		\boldsymbol{b}_{1}\right) \\
		&&\times \func{exp}\left( -\frac{1}{2\sigma _{0}^{2}}\sum_{t=1}^{T}\left(
		\left( y_{t}-m_{t}\right) ^{2}-y_{t}^{2}\right) \right) dQ_{1}\left( 
		\boldsymbol{b}_{1}\right) dJ\left( \psi \right) +o_{p}\left( 1\right) ,
	\end{eqnarray*}%
	and then by (\ref{eq:rtyt}) 
	\begin{equation*}
		LR_{1}=\underset{:=LR_{2}}{\underbrace{\int \int_{\left\vert \boldsymbol{b}%
					_{1}\right\vert <\log p}\func{exp}\left( -\frac{1}{2\sigma _{0}^{2}}%
				\sum_{t=1}^{T}\left( \left( y_{t}-m_{t}\right) ^{2}-y_{t}^{2}\right) \right)
				dQ_{1}\left( \boldsymbol{b}_{1}\right) dJ\left( \psi \right) }}%
		+o_{p}\left( 1\right) .
	\end{equation*}
	
	Next, let $M=Ex_{t}x_{t}^{\prime },\ $and 
	\begin{equation*}
		{\mathcal{H}}\left( \psi \right) =T\otimes M,\ {\text{with }}T=-\frac{1}{%
			\sigma _{0}^{2}}\left( 
		\begin{array}{cc}
			1-\psi & 1-\psi \\ 
			1-\psi & 1%
		\end{array}%
		\right) .
	\end{equation*}%
	Also let $x_{t}\left( \psi \right) =\left( x_{t}^{\prime }1\left\{
	t/T>\psi \right\} ,x_{t}^{\prime }\right) ^{\prime },$ $h=p^{-1/4}\left(
	1,\psi -1\right) ^{\prime }\otimes \boldsymbol{b}_{1}$, and 
	\begin{equation*}
		\bar{\theta}_{\psi }={\mathcal{H}}\left( \psi \right) ^{-1}\frac{1}{%
			\sqrt{T}\sigma _{0}^{2}}\sum_{t=1}^{T}x_{t}\left( \psi \right) \nu
		_{t}.
	\end{equation*}%
	As $\psi _{t}=\psi -1\left\{ t/T\leq \psi \right\} ,$ we can write 
	\begin{eqnarray*}
		&&-\frac{1}{2\sigma _{0}^{2}}\sum_{t=1}^{T}\left( -2\nu _{t}\left( T%
		\sqrt{p}\right) ^{-1/2}\psi _{t}x_{t}^{\prime }\boldsymbol{b}_{1}+\left( T%
		\sqrt{p}\right) ^{-1}\left( \psi _{t}x_{t}^{\prime }\boldsymbol{b}%
		_{1}\right) ^{2}\right) , \\
		&=&\frac{1}{\sqrt{T}\sigma _{0}^{2}}\sum_{t=1}^{T}\nu
		_{t}x_{t}\left( \psi \right) ^{\prime }h-h^{\prime }\frac{1}{2T\mathfrak{\
				\sigma }_{0}^{2}}\sum_{t=1}^{T}\left( x_{t}\left( \psi \right) x_{t}\left(
		\psi \right) ^{\prime }\right) h \\
		&=&h^{\prime }{\mathcal{H}}\left( \psi \right) \bar{\theta}-\frac{1}{2}%
		h^{\prime }{\mathcal{H}}\left( \psi \right) h+o_{p}\left( 1\right) \\
		&=&\frac{1}{2}\bar{\theta}^{\prime }{\mathcal{H}}\left( \psi \right) \bar{%
			\theta}-\frac{1}{2}\left( \bar{\theta}-h\right) ^{\prime }{\mathcal{H}}%
		\left( \psi \right) \left( \bar{\theta}-h\right) +o_{p}\left( 1\right) ,
	\end{eqnarray*}%
	by the functional CLT for the second equality, which is uniform in $\psi $
	and $\left\vert \boldsymbol{b}_{1}\right\vert <\log p$, and by a simple
	algebra for the third equality. Thus, solving squares and recalling $%
	y_{t}=\sigma _{0}\nu _{t}$ implies that $LR_{2}$ equals 
	\begin{eqnarray*}
		&&\int \int_{\left\vert \boldsymbol{b}_{1}\right\vert <\log p}\func{exp}%
		\left( \frac{1}{2}\bar{\theta}^{\prime }{\mathcal{H}}\left( \psi \right) 
		\bar{\theta}-\frac{1}{2}\left( \bar{\theta}-h\right) ^{\prime }{\mathcal{H}}%
		\left( \psi \right) \left( \bar{\theta}-h\right) \right) dQ_{1}\left( 
		\boldsymbol{b}_{1}\right) dJ\left( \psi \right) +o_{p}\left( 1\right) \\
		&=&\underset{LR_{3}}{\underbrace{\int \int \func{exp}\left( \frac{1}{2}\bar{%
					\theta}^{\prime }{\mathcal{H}}\left( \psi \right) \bar{\theta}-\frac{1}{2}%
				\left( \bar{\theta}-h\right) ^{\prime }{\mathcal{H}}\left( \psi \right)
				\left( \bar{\theta}-h\right) \right) dQ_{1}\left( \boldsymbol{b}_{1}\right)
				dJ\left( \psi \right) }}+o_{p}\left( 1\right)
	\end{eqnarray*}%
	since 
	\begin{equation*}
		\func{exp}\left( \frac{1}{2}\bar{\theta}^{\prime }{\mathcal{H}}\left( \psi
		\right) \bar{\theta}-\frac{1}{2}\left( \bar{\theta}-h\right) ^{\prime }{\ 
			\mathcal{H}}\left( \psi \right) \left( \bar{\theta}-h\right) \right) \leq 
		\func{exp}\left( \frac{1}{2}\bar{\theta}^{\prime }{\mathcal{H}}\left( \psi
		\right) \bar{\theta}\right) =O_{p}\left( 1\right)
	\end{equation*}%
	and $\int \int_{\left\vert \boldsymbol{b}_{1}\right\vert \geq \log
		p}dQ_{1}\left( \boldsymbol{b}_{1}\right) dJ\left( \psi \right) \rightarrow
	0$. Furthermore, it follows from the algebra of \cite{Andrews1994a} in their
	Theorem A.1 (b) that 
	\begin{eqnarray*}
		&&\int \func{exp}\left( \frac{1}{2}\bar{\theta}^{\prime }{\mathcal{H}}\left(
		\psi \right) \bar{\theta}-\frac{1}{2}\left( \bar{\theta}-h\right) ^{\prime
		}{\mathcal{H}}\left( \psi \right) \left( \bar{\theta}-h\right) \right)
		dQ\left( h\right) \\
		&=&\left( 1+\mathfrak{c}_{p}\right) ^{-p/2}\func{exp}\left( \frac{1}{2}\frac{%
			\mathfrak{c}_{p}}{1+\mathfrak{c}_{p}}\psi \left( 1-\psi \right) \frac{1}{%
			\sigma _{0}^{2}}\bar{\kappa}_2\left( \psi \right) ^{\prime }M\bar{\kappa}_{2}%
		\left( \psi \right) \right)
	\end{eqnarray*}%
where $\mathfrak{c}_{p}=\mathfrak{c}/\sqrt{p},$ for which we recall that $%
	h=p^{-1/4}\left( 1,\psi -1\right) ^{\prime }\otimes \boldsymbol{b}_{1}$
	and $\boldsymbol{b}_{1}\sim N\left( 0,\mathfrak{c\sigma }%
	_{0}^{2}\left( \psi \left( 1-\psi \right) M\right) ^{-1}\right) $. Also,
	due to Theorem \ref{thm:local_power}, 
	\begin{equation*}
		W_{T}\left( \psi \right) =\psi \left( 1-\psi \right) \frac{1}{\sigma
			_{0}^{2}}\bar{\kappa}_{2}\left( \psi \right) ^{\prime }M\bar{\kappa}_{2}\left(
		\psi \right) +o_{p}\left( \sqrt{p}\right)
	\end{equation*}%
	uniformly in $\psi .$ Thus, 
	\begin{equation*}
		LR_{3}=\underset{:=ExpW_{T}}{\underbrace{\left( 1+\mathfrak{c}/\sqrt{p}%
				\right) ^{-\frac{p}{2}}\int \func{exp}\left( \frac{1}{2}\frac{\mathfrak{c}/%
					\sqrt{p}}{1+\mathfrak{c}/\sqrt{p}}W_{T}\left( \psi \right) \right)
				dJ\left( \psi \right) }}+o_{p}\left( 1\right) .
	\end{equation*}%
	This completes the asymptotic equivalence under the null.
	
	Turning to the asymptotic equivalence under the alternatives $\mathcal{H}%
	_{\ell T},$ we do it by resorting to the contiguity of the local alternative
	densities to the null density. That is, if the two tests are equivalent
	under the null and the local alternatives are contiguous to the null, the
	probability that they are not equal should converge to zero under these
	alternatives as well by the definition of the contiguity. Since the
	likelihood ratio statistic converges to $\int \func{exp}\left( \frac{%
		\mathfrak{c}}{\sqrt{2}}\mathcal{Z}\left( \psi \right) -\frac{\mathfrak{c}%
		^{2}}{4}\right) dJ\left( \psi \right) $ under the null $\mathcal{H}_{0}$
	with $\mathcal{Z}\left( \psi \right) $ distributed as the standard normal,
	we note that 
	$	E\func{exp}\left( \frac{\mathfrak{c}}{\sqrt{2}}\mathcal{Z}\left( \psi
		\right) -\frac{\mathfrak{c}^{2}}{4}\right)  =1$, 
        from which the contiguity follows by le Cam's first lemma, see e.g. \cite%
	{VanderVaart1998}. This completes the proof.
\end{proof}



\section{Proofs of remaining results}

 Denote $L(\psi )=\left( X^{\ast }(\psi )^{\prime }M_{X}X^{\ast }(\psi )\right)
	^{-1}X^{\ast }(\psi )^{\prime }M_{X}$
where $X^{\ast }(\psi )$ has $t$-th row $x_{t}^{\ast }(\psi )^{\prime
}=x_{t}^{\prime }1\left\{ t/T>\psi \right\} $, $M_{X}$ the residual
maker for the matrix $X$ with $t$-th row $x_{t}^{\prime }$, and $
	F(\psi )=SP(\psi )^{-1}\Xi (\psi )P(\psi )^{-1}S^{\prime }.$
Set $\bar{\Xi}(\psi )=T^{-1}\sum_{t=1}^{T}x_{t}(\psi
)x_{t}^{\prime }(\psi )\sigma _{t}^{2}$ and $x_{t}(\psi )=\left(
x_{t}^{\prime },x_{t}^{\prime \ast }(\psi )\right) ^{\prime }$ and define $\hat{H}(\psi
)=T^{-1}X^{\prime \ast }(\psi )X(\psi )$.
\begin{proof}[Proof of Theorem \ref{thm:null_Chow}:]
The first step of the proof is to show that 
\begin{equation} \label{disteq1}
{\mathcal{Z}}_T(\psi )={\mathcal{C}}
	_T(\psi )/{\sqrt{2p}}+o_p(1) ,
\end{equation} 
uniformly in $\psi\in\Psi$, where
\begin{equation}
	{\mathcal{C}}_{T}(\psi )=\frac{T^{-1}\sum_{s\neq t}\upchi_{t}(\psi )^{\prime
		}\Xi ^{-1}\upchi_{s}(\psi )\nu _{t}\nu _{s}}{\psi \left(
		1-\psi \right) \sqrt{2p}},  
\end{equation}
and $\upchi_{t}(\psi )=x_t 1\left\{t/T\leq\psi\right\}-\psi x_t$. 
After establishing (\ref{disteq1}), the claim would follow by a FCLT for $ \mathcal{C}_{T}(\psi )$. Thus, we first prove (\ref{disteq1}) uniformly in $\psi\in\Psi$. 

We proceed by first establishing the approximation 
\begin{equation}\label{disteq2} 
{\mathcal{Z}}_T(\psi )=\left({{\mathcal{L}}
			_T(\psi )-p}\right)/{\sqrt{2p}}+o_p(1),
\end{equation}            
            uniformly in $\psi\in\Psi$, where $\mathcal{L}
		_T(\psi )$ is defined in (\ref{Rdef}).

To show tightness of the process, note that ${\mathcal{S}}_{T}(\psi )$ equals $\left[ \psi \left(
1-\psi \right) \sqrt{2p}\right] ^{-1}$ times 
\begin{equation*}
	\frac{1}{T}\underset{s\neq t}{\sum_{s,t=1}^{[T\psi ]}}x_{t}^{\prime
	}\Xi ^{-1}x_{s}\nu _{t}\nu _{s}-\frac{2\psi }{T}%
	\underset{s\neq t}{\sum_{s=1}^{T}\sum_{t=1}^{[T\psi ]}}x_{t}^{\prime
	}\Xi ^{-1}x_{s}\nu _{t}\nu _{s}+\frac{\psi ^{2}}{T}%
	\underset{s\neq t}{\sum_{s,t=1}^{T}}x_{t}^{\prime }\Xi
	^{-1}x_{s}\nu _{t}\nu _{s}
\end{equation*}%
and thus 
\begin{align*}
	{\mathcal{S}}_{T}(\psi )& =\frac{\sqrt{2}}{\psi \left( 1-\psi \right) }%
	\left[ \mathcal{A}_{T}(\psi )-\psi \left[ \mathcal{A}_{T}(1)+\mathcal{A}_{T}(\psi )-\bar{\mathcal{A}}%
	_{T}(\psi )\right] +\psi ^{2}\mathcal{A}_{T}(1)\right] , \\
	& =\sqrt{2}\left( \frac{\mathcal{A}_{T}(\psi )}{\psi }+\frac{\bar{\mathcal{A}}_{T}\left(
		\psi \right) }{\left( 1-\psi \right) }-\mathcal{A}_{T}(1)\right) ,
\end{align*}%
where 
\begin{align*}
	\mathcal{A}_{T}(\psi )& =\frac{1}{T\sqrt{p}}{\sum_{s=2}^{[T\psi ]}}{%
		\sum_{t=1}^{s-1}}\upsilon _{t}^{\prime }\upsilon _{s}, \\
	\bar{\mathcal{A}}_{T}(\psi )& =\frac{1}{T\sqrt{p}}{\sum_{s=[T\psi
			]+1}^{T}\sum_{t=[T\psi ]+1}^{s-1}}\upsilon _{t}^{\prime }\upsilon _{s},
\end{align*}%
and $\upsilon _{t}=\left\{ \upsilon _{ti}\right\} _{i=1}^{p}=\Xi
^{-1/2}x_{t}\nu _{t}$ being an mds. The tightness
of $\mathcal{A}_{T}\left(\psi\right)$ is established in Section A.2.2 of \cite{Gupta2023}. Due to their symmetric nature, the tightness proof
is almost the same for $\bar{\mathcal{A}}_{T}(\psi )$. Finally, the limiting covariance kernel is derived in Section A.2.3 of \cite{Gupta2023}.

\end{proof}
\begin{proof}[Proof of Theorem \ref{thm:local_power}:]
	We present a general proof where the true break point is $\psi_0$, and setting $\psi=\psi_0$ gives our claim in the paper. Under $\mathcal{H}_{\ell }$, we have $\hat{\kappa}_{2}(\psi )=L(\psi )X^{\ast
		}\left( \psi _{0}\right) \kappa _{2\ell }+L(\psi )\nu +L(\psi
		)r$, so that, writing $A\left( \psi ,\psi _{0}\right) =L(\psi )X^{\ast
		}\left( \psi _{0}\right) $ and $\hat{F}(\psi )=S\hat{P}(\psi )^{-1}\hat{\Xi}(\psi )\hat{P}(\psi )^{-1}S^{\prime }$, similar algebra to that used earlier
		yields 
		\begin{eqnarray}
			{\mathcal{Z}}_{T}(\psi ) &=&{{\mathcal{C}}_{T}(\psi )}+\frac{2T\kappa
				_{2\ell }^{\prime }A\left( \psi ,\psi _{0}\right) ^{\prime }\hat{F}%
				(\psi )^{-1}L(\psi )\nu }{\sqrt{2p}}+\frac{2T\kappa _{2\ell
				}^{\prime }A\left( \psi ,\psi _{0}\right) ^{\prime }\hat{F}(\psi
				)^{-1}L(\psi )r}{\sqrt{2p}}  \notag \\
			&+&\frac{T\kappa _{2\ell }^{\prime }A\left( \psi ,\psi _{0}\right)
				^{\prime }\left( \hat{F}(\psi )^{-1}-{F}(\psi )^{-1}\right) D\left(
				\psi ,\psi _{0}\right) \kappa _{2\ell }}{\sqrt{2p}}  \label{local_alt_1} \\
			&+&\frac{T\kappa _{2\ell }^{\prime }A\left( \psi ,\psi _{0}\right)
				^{\prime }\hat{F}(\psi)^{-1}A\left( \psi ,\psi _{0}\right) \kappa _{2\ell }}{\sqrt{2p}}%
			+o_{p}(1). \notag
		\end{eqnarray}%
		For the second term on the RHS of (\ref{local_alt_1}), note that this equals 
		\begin{eqnarray}
			&&\frac{2T\kappa _{2\ell }^{\prime }A\left( \psi ,\psi _{0}\right)
				^{\prime }{F}(\psi )^{-1}L(\psi )\nu }{\sqrt{2p}}+\frac{2T\kappa
				_{2\ell }^{\prime }A\left( \psi ,\psi _{0}\right) ^{\prime }\left( \hat{F%
				}(\psi )^{-1}-{F}(\psi )^{-1}\right) L(\psi )\nu }{\sqrt{2p}} 
			\notag  \label{local_alt_2} \\
			&\leq&\frac{2T\kappa _{2\ell }^{\prime }A\left( \psi ,\psi _{0}\right) {F}%
				(\psi )^{-1}L(\psi )\nu }{\sqrt{2p}}+O_{p}\left(\frac{1}{\sqrt{p}}\right)\mu_T^{-2} n\left\Vert
			\kappa _{2\ell }\right\Vert \left\Vert T^{-1}X^{\prime }\nu
			\right\Vert \left\Vert \hat{F}(\psi )-{F}(\psi )\right\Vert %
			 \notag \\
			&=&\frac{2T\kappa _{2\ell }^{\prime }A\left( \psi ,\psi _{0}\right)
				^{\prime }{F}(\psi )^{-1}L(\psi )\nu }{\sqrt{2p}}+O_{p}\left( 
			\mu_T^{-4}p^{1/4} \max \left\{ \mu_T^{-1}\tau _{p},\zeta_p \right\} \right) ,  \notag
		\end{eqnarray}%
		\sloppy the second stochastic order above
		being negligible by (\ref{rate:Q_weak_conv}). By Assumption \ref{ass:errors}%
		, the first term  has mean zero and variance equal to
		a constant times $
			{\varrho ^{\prime }A\left( \psi ,\psi _{0}\right) ^{\prime } F(\psi)^{-1}L(\psi
				)L(\psi)'F(\psi )^{-1}A\left( \psi ,\psi _{0}\right) \varrho }/{\sqrt{p}}=O_{p}\left(1/
			\sqrt{p}\right),$
		uniformly in $\psi $ by Lemmas \ref{lemma:B_eigs} and \ref{lemma:Bhat_Torm}
		and the calculations therein.
		
		By Assumption \ref{ass:aprx0}, the third term on the RHS of (\ref{local_alt_1}) is $
			O_{p}\left( n\left\Vert \kappa _{2\ell }\right\Vert \left\Vert
			T^{-1}X^{\prime }r\right\Vert /\sqrt{p}\right) =O_{p}\left( p^{-1/4}%
			\right) .$ The fourth term on the RHS of (\ref{local_alt_1}) is readily seen to be bounded by $%
		O_{p}(1) \left\Vert \hat{F}(\psi )^{-1}-{F}(\psi )^{-1}\right\Vert
		 =O_{p}\left( \mu_T^{-4}\max \left\{\mu_T^{-1} \tau _{p},\zeta_p \right\} \right) 
		$, which is negligible by (\ref{rate:Q_weak_conv}). Thus, using similar steps to replace $\hat{F}(\psi )^{-1}$ by ${F}(\psi )^{-1}$ in the fifth term on the RHS and by (\ref{B_gamma1}%
		), (\ref{local_alt_1}) becomes 
		\begin{eqnarray}
			{\mathcal{Z}}_{T}(\psi ) &=&{\mathcal{C}}_{T}(\psi )+\frac{T\kappa
				_{2\ell }^{\prime }A\left( \psi ,\psi _{0}\right) ^{\prime }{F}(\psi
				)^{-1}A\left( \psi ,\psi _{0}\right) \kappa _{2\ell }}{\sqrt{2p}}%
			+o_{p}(1)  \notag \\
			&=&{{\mathcal{C}}_{T}(\psi )}+\psi (1-\psi ){\varrho ^{\prime }A\left(
				\psi ,\psi _{0}\right) ^{\prime }P \Xi^{-1} P A\left( \psi ,\psi _{0}\right)
				\varrho }+o_{p}(1).  \notag
		\end{eqnarray}%
		\sloppy Now, by the definition of its components and steps similar to those
		elsewhere in the paper, it is readily seen that 
		$\left\Vert A\left( \psi ,\psi _{0}\right) -\left\{{\left( \psi +\psi
			_{0}(1-\psi )-\max\left\{\psi, \psi _{0}\right\} \right) }/{\psi
			(1-\psi )}\right\}I_{p}\right\Vert =o_{p}(1),
		$ uniformly on $\Psi $ and that $\psi +\psi _{0}(1-\psi )-\max\left\{\psi, \psi _{0}\right\} =-\psi \psi _{0}+ \min\left\{\psi, \psi _{0}\right\} $ as $\psi +\psi _{0}-\max\left\{\psi, \psi _{0}\right\} =\min\left\{\psi, \psi _{0}\right\} .$ Thus, 
		\begin{equation}
			{\mathcal{Z}}_{T}(\psi ){\Rightarrow} Q(\psi )+\frac{\left( \psi \psi
				_{0}-\min\left\{\psi, \psi _{0}\right\} \right) ^{2}}{\psi (1-\psi )%
			}\lim_{T\rightarrow \infty }{\varrho ^{\prime}P \Xi^{-1} P\varrho },  \label{local_alt_3}
		\end{equation}%
		by Theorem \ref{thm:null_Chow}, which
		gives the weak limit of ${\mathcal{Z}}_{T}(\psi )$ under $\mathcal{H}_{\ell }$.
	\end{proof}

\begin{proof}[Proof of Theorem \ref{thm:bootstrap}:]
It is sufficient to check (\ref{bootthmorig}), whence (\ref{bootthmnew}) follows. Let $ \nu^{\star} $ denote the $T\times 1$ vector with elements  $ \nu_t^{\star} = \hat{u}_t(\psi) \upsilon_t  $, where $ \upsilon_t $ is an iid sequence of Rademacher variables. Then, $\hat{\kappa}_{2}^{\star}(\psi )= L(\psi )\nu^{\star}, $
	since $\kappa _{2}=0$ under $\mathcal{H}_{0}$. Also, we have 
	\begin{equation}
		W_{T}^{\star}(\psi )=T\left( \nu^{\star} \right) ^{\prime }L^{\prime }(\psi )%
		\hat{F}^{\star}(\psi )^{-1}L(\psi ) \nu^{\star} ,
		\label{wald_*}
	\end{equation}%
	where $\hat{F}^{\star}(\psi )=S\hat{P}(\psi )^{-1}\hat{\Xi}^{\star}(\psi )\hat{P}(\psi )^{-1}S^{\prime }$ and $ \hat{\Xi}^{\star}(\psi ) $ is constructed in similar fashion to $ \hat{\Xi}(\psi ) $ but with the bootstrap sample. Writing $ \bar{W}_{T}^{\star}(\psi )=T\left( \nu^{\star} \right) ^{\prime }L^{\prime }(\psi )%
	{F}(\psi )^{-1}L(\psi ) \nu^{\star} $, observe that $E^{\star} \bar{W}_{T}^{\star}(\psi )$ is \begin{equation}  \nonumber 
		T tr L^{\prime }(\psi )		\hat{F}(\psi )^{-1}L(\psi ) E^{\star} \nu^{\star} (\nu^{\star})' 
		= T tr L^{\prime }(\psi )		\hat{F}(\psi )^{-1}L(\psi ) diag\left[\hat{u}_1(\psi)^2,...,\hat{u}_T(\psi)^2\right],
\end{equation}
which is $ p + o_p \left(p^{1/2}\right) $ uniformly in $ \psi $, due to Lemma \ref{lemma:deltahat}.  
	
	\sloppy We next show that $ E^{\star} \bar{W}_{T}^{\star}(\psi ) - E^{\star} W_{T}^{\star}(\psi )= o_p \left(p^{1/2}\right) $. Indeed, 
$		E^{\star}\left\vert \bar{W}_{T}^{\star}(\psi ) -  W_{T}^{\star}(\psi )\right\vert \leq   
		E^{\star} \left( \left\Vert T^{-1/2}L^{\prime }(\psi)\nu^{\star} \right\Vert
		^{2}\left\Vert \hat{F}(\psi )^{-1}-\hat{F}^{\star}(\psi )^{-1}\right\Vert \right).$
To bound $ E^{\star}\left\Vert \hat{F}(\psi )^{-1}-\hat{F}^{\star}(\psi )^{-1}\right\Vert^2	 $, we proceed by deriving bounds for $ E^{\star}\left\Vert \hat{F}^{\star}(\psi )^{-1}\right\Vert^4	 $ and $ E^{\star} \left\Vert \hat{F}(\psi )-\hat{F}^{\star}(\psi )\right\Vert^4 $, and applying Cauchy-Schwarz inequality. 
	Recall $ \hat{F}(\psi )-\hat{F}^{\star}(\psi ) = S' \hat{P}(\psi)^{-1} \left(\hat{\Xi}(\psi )-\hat{\Xi}^{\star}(\psi )\right)\hat{P}(\psi)^{-1}S $, and also $ \sup_{\psi \in \Psi} \left\Vert \hat{P}(\psi)^{-1} \right\Vert = O_p \left(\mu_T^{-1}\right)$ by Lemma \ref{lemma:Mhat_Torm}. 
	Following the proof of Lemma \ref{lemma:Omega_hat_true},  $ \hat{\Xi}(\psi )-\hat{\Xi}^{\star}(\psi )$ is the sum of $ \upalpha_1^{\star}(\psi) $ and $\upalpha_3^{\star} (\psi) $ therein. By the triangle and $ c_r $ inequalities, we need only show $ E^{\star} \left\Vert \upalpha_j^{\star}(\psi) \right\Vert^4 = O_p (\mu_T^{-8} p^{12}/T^{4}) $, for $ j=1,3 $. Note that by independence of the sequence $ \upsilon_{t} $
	\begin{eqnarray*}
		E^{\star} \left\Vert \upalpha_{1}^{\star}(\psi )\right\Vert^4 &\leq& \left(T^{-1}\sum_{t=1}^{T}\left(
		x_{t}^{\prime }(\psi )x_{t}(\psi )\right) ^{2} \right)^4 
		E^{\star} \left(\left( \kappa^{\star} -\hat{\kappa}^{\star}(\psi )\right) ^{\prime }\left( \kappa^{\star} -\hat{\kappa}^{\star}(\psi )\right) \right)^4 \\
		&\leq & 
		O_p (p^8) \left\Vert \hat{P}(\psi )^{-1}\right\Vert^8 T^{-8}\sum_{t_1,t_2,t_3,t_4} \hat{u}_{t_1}^2 x_{t_1}'x_{t_1} \cdots \hat{u}_{t_4}^2 x_{t_4}'x_{t_4}, 
	\end{eqnarray*} 
	as desired, with the bound for $ \upalpha_3^{\star} $ similarly obtained. Combine these results to get $ E^{\star}\left\Vert \hat{F}(\psi )^{-1}-\hat{F}^{\star}(\psi )^{-1}\right\Vert^2	=O_p (\mu_T^{-10} p^{12}/T^{4}) $. The bound for $ E^{\star}\left\Vert \hat{F}^{\star}(\psi )^{-1}\right\Vert^4	 $ follows in much the same manner, observing only the similarity for the derivations for the sample counterparts in Lemmas  \ref{lemma:Omega_hat_true} and \ref{lemma:Bhat_Torm}. Next, similar to steps used to obtain previous bounds, $
		E^{\star} \left\Vert T^{-1/2}A^{\prime }(\psi)\nu^{\star} \right\Vert
		^{4} = 
		O_p\left(\mu_T^{-4}\right) \left( T^{-1} \sum_{t=1}^T x_t 'x_t \hat{u}_t^2 \right)^2 = 
		O_p (\mu_T^{-4} p^2),$
	as $ \upsilon_t $ is an iid Rademacher sequence. 
	Then, under the condition \eqref{rate:Q_weak_conv}, $ \mu_T^{-14} p^{14}/T^{4} = o (p)$ and this completes the proof.
\end{proof}
\begin{proof}[Proof of Theorem \ref{theorem:exptest_dist}:]
By Theorems \ref{thm:null_Chow}, \ref{thm:local_power} \ref{thm:bootstrap}, the claim follows by establishing the weak convergence of $\hat {\upomega}$ and the continuous mapping theorem. From Lemma 1 (c) of \cite{Sun2014}, weak convergence of $\hat {\upomega}$ follows if 
		\begin{equation}
			\hat{\upomega}-\tilde{\upomega}=o_p(1) \label{eq:VhatVtil},
		\end{equation} 
		where $\tilde{\upomega}=\frac{2}{T}\text{\ensuremath{\sum_{t=2}^{T}\text{\ensuremath{\sum_{s=2}^{T}k\left(\frac{t-s}{T/h}\right)\bar{\ell}^{\star}_{s}}}\bar{\ell}^{\star}_{t}}},$ $ {\ell}^{\star}_t = \left( Tp\right)^{-1/2} x_{t}'{\Xi}^{-1} {\nu}_{t} \text{\ensuremath{\sum_{s=1}^{t-1}x_{s}{\nu}_{s}}} $,  $\bar{\ell}^{\star}_{t}= {\ell}^{\star}_t - T^{-1}\sum_{t=2}^{T} {\ell}^{\star}_t$. Sun's Lemma 1 (c) approximates the partial sums of $ \ell_{t}$ by the partial sums of $ e_t $, which is iid normal, but the argument also holds when it is approximated by the partial sums of $ a_{Tt }e_t $ for any real bounded array $ a_{Tt} $. In the present case, $ a_{Tt} = \sqrt{t/T} $. 
		
Let $\upupsilon=T/h$, $\hat\varkappa_t=\hat{\upxi}_t'\sum_{s<t}\hat{\upxi}_s/\sqrt{p}=\sqrt{T}\ell_t$, $\hat{\upxi}_t=\hat\Xi^{-1/2}x_t\hat{u}_t$, $\bar{\hat\varkappa}/\sqrt{T}=T^{-1}\sum_{t=2}^{T} \ell_t=T^{-1}\sum_{t=2}^{T}\hat\varkappa_t/\sqrt{T}$,  with analogous definitions using $\Xi$ and $\nu_t$ for $\varkappa_t$, $\upxi_t$ and $\bar{\varkappa}$. Then \begin{eqnarray}
	\hat{{\upomega}}-\tilde{{\upomega}}&=&T^{-2}\sum_{j=-(n-1)}^{T-1}k\left(j/\upupsilon \right)T^{-1}\sum_{t=1+\max\left\{j,0\right\}}^{T-\left\vert \min\left\{j,0\right\}\right\vert}\left\{\left(\hat\varkappa_t\hat\varkappa_{t+|j|}-\varkappa_t\varkappa_{t+|j|}\right)+2\bar{\hat\varkappa}\left(\hat\varkappa_t-\varkappa_t\right)\right.\nonumber\\
	&+&\left.\left(\bar{\hat\varkappa}-\bar\varkappa\right)\hat\varkappa_t+\left(\bar{\hat\varkappa}^2-\bar\varkappa^2\right)\right\}.    \label{Vhat_Vtilde1}
\end{eqnarray}
We show the bound for
\begin{equation}\label{zetapdiff}
	\hat\varkappa_t\hat\varkappa_{t+|j|}-\varkappa_t\varkappa_{t+|j|}=\left(\hat\varkappa_t-\varkappa_t\right)\hat\varkappa_{t+|j|}+\left(\hat\varkappa_{t+|j|}-\varkappa_{t+|j|}\right)\hat\varkappa_{t},
\end{equation}
while omitting similar details for the other three terms for the sake of brevity.
First note that $\left\Vert\hat{\upxi}_t\right\Vert=O_p\left( \left\Vert x_t\right\Vert\right)=O_p\left( \sqrt{p}\right)$, 
by Assumption \ref{ass:M_diff}$(ii)$, Assumption \ref{ass:M_diff}$(i)$, and
\begin{equation}
	\hat{u}_{t}=y_{t}-x_{t}^{\prime }\hat{\kappa _{1}}(\psi
	)=x_{t}^{\prime }\left( \hat{\kappa}_{1}(\psi )-\kappa _{1}\right)
	+x_{t}^{\prime }1\left( t/T>\psi \right) \kappa _{2\ell
	}+r_{t}+\nu _{t}=O_{p}(1).  \label{epsilonhatorder}
\end{equation} Hence
\begin{equation}\label{zetabound}
	\hat\varkappa_t=\hat{\upxi}_t'\sum_{s<t}\hat{\upxi}_s/\sqrt{p}=O_p\left(T\sqrt{p}\right).
\end{equation}
By the same argument, $\left\Vert \upxi_t\right\Vert=O_p(\sqrt{p})$ and $\varkappa_t=O_p\left(T\sqrt{p}\right)$  as well. 

Now recall that $\hat{u}_{t}-\nu _{t}=x_{t}^{\prime }\left( 
\hat{\kappa}_{1}(\psi )-\kappa _{1}\right) +x_{t}^{\prime }1\left(
t/T>\psi \right) \kappa _{2\ell }+r_{t}$ and $\left\Vert \hat{\kappa}%
_{1}(\psi )-\kappa _{1}\right\Vert \leq  \left\Vert \hat{\kappa}%
(\psi )-\kappa \right\Vert =O_{p}\left(\mu_T^{-1} \sqrt{p}/\sqrt{T}\right) $
so that 
\begin{equation}
	\hat{u}_{t}-\nu _{t}=O_{p}\left( \max \left\{ \mu_T^{-1}p/\sqrt{T}%
	,p^{3/4}/\sqrt{T}\right\} \right).  \label{epsdifforder}
\end{equation}
		Thus 
		\begin{equation}\label{hdiffbound}
			\left\Vert\hat{\upxi}_t-\upxi_t\right\Vert=\left\Vert\Xi^{-1}\left(\Xi-\hat\Xi\right)\hat\Xi^{-1}x_t\hat{u}_t+\Xi^{-1}x_t\left(\hat{u}_t-\nu_t\right)\right\Vert=O_p\left(\mu_T^{-2}\sqrt{p}\max\left\{\zeta_p,p/\sqrt{T}\right\}\right),
		\end{equation}
		using Assumption \ref{ass:M_diff}$(iii)$. Using (\ref{hdiffbound}), 
		\begin{equation}\label{zetadiffbound}
			\hat\varkappa_t-\varkappa_t =\left(\hat{\upxi}_t-\upxi_t\right)'\sum_{s<t}\hat{\upxi}_s/\sqrt{p}+\hat{\upxi}_t'\sum_{s<t}\left(\hat{\upxi}_s-\upxi_s\right)/\sqrt{p}=O_p\left(\mu_T^{-2}T\sqrt{p}\max\left\{\zeta_p,p/\sqrt{T}\right\}\right).
		\end{equation}  
		Then we conclude that $\hat\varkappa_t\hat\varkappa_{t+|j|}-\varkappa_t\varkappa_{t+|j|}=O_p\left(\mu_T^{-2}T^2p\max\left\{\zeta_p,p/\sqrt{T}\right\}\right)$ upon applying (\ref{zetabound}) and (\ref{zetadiffbound}) in (\ref{zetapdiff}). This, along with similar steps to derive bounds for the remaining terms in (\ref{Vhat_Vtilde1}) and Lemma 1 of \cite{Jansson2002}, yields
		\[
		\hat{{\upomega}}-\tilde{{\upomega}} =O_{p}\left(\upupsilon \left( \int_{%
			{\mathbb{R}}}\left\vert k(x)\right\vert dx\right) \mu_T^{-2}p\max\left\{\zeta_p,p/\sqrt{T}\right\}
		\right) =O_{p}\left(\mu_T^{-2} \max\left\{p\zeta_p,p^2/\sqrt{T}\right\} \right) ,
		\]
		which is negligible under the conditions of the theorem.
\end{proof}

\begin{proof}[Proof of Proposition \ref{prop:seq}]
	\cite{Andrews1993} shows in (A.33) therein that the process $\mathcal{W}_{p}$ can be
	equivalently represented as 
	\begin{equation*}
		\mathcal{W}_{p}\left( \psi \right) =\frac{B_{p}\left( \psi /\left(
			1-\psi \right) \right) ^{\prime }B_{p}\left( \psi /\left( 1-\psi
			\right) \right) }{\psi /\left( 1-\psi \right) }=\frac{B_{p}\left( \phi
			\right) ^{\prime }B_{p}\left( \phi \right) }{\phi }=:\bar{\mathcal{W}}%
		_{p}\left( \phi \right) ,
	\end{equation*}%
	where $\phi =\phi \left( \psi \right) =\psi /\left( 1-\psi \right) $,
	which is strictly increasing in $\psi \in \left( 0,1\right) $.
	
	Then, we show that the entire covariance kernel of $\bar{\mathcal{W}}%
	_{p}\left( \phi \right) $ is proportional to $2p$ and pivotal. Let $\phi
	_{1}<\phi _{2}$ and $\Pi =B_{p}\left( \phi _{2}\right) -B_{p}\left( \phi
	_{1}\right) $. Note the independent increment condition of the Brownian
	motion $B_{p}$ and that $B_{p}\left( \phi \right) ^{\prime }B_{p}\left( \phi
	\right) /\phi $ and $\Pi ^{\prime }\Pi /\left( \phi _{2}-\phi
	_{1}\right) $ are independent and Chi-squared distributed with mean $p$ and
	variance $2p$ to obtain that 
	\begin{eqnarray*}
		&&E\bar{\mathcal{W}}_{p}\left( \phi _{1}\right) \bar{\mathcal{W}}_{p}\left(
		\phi _{2}\right) \\
		&=&E\left( \frac{B_{p}\left( \phi _{1}\right) ^{\prime }B_{p}\left( \phi
			_{1}\right) }{\phi _{1}}\right) \left( \frac{B_{p}\left( \phi _{1}\right)
			^{\prime }B_{p}\left( \phi _{1}\right) }{\phi _{2}}+\frac{\Pi ^{\prime
			}\Pi }{\phi _{2}}+2\frac{B_{p}\left( \phi _{1}\right) ^{\prime }\Pi }{%
			\phi _{2}}\right) \\
		&=&\frac{E\left( B_{p}\left( \phi _{1}\right) ^{\prime }B_{p}\left( \phi
			_{1}\right) \right) ^{2}}{\phi _{1}\phi _{2}}+\frac{E\left( B_{p}\left( \phi
			_{1}\right) ^{\prime }B_{p}\left( \phi _{1}\right) \right) E\Pi ^{\prime
			}\Pi }{\phi _{1}\phi _{2}} \\
		&=&\phi _{1}\frac{2p+p^{2}}{\phi _{2}}+\frac{p^{2}}{\phi _{2}}\left( \phi
		_{2}-\phi _{1}\right) 
		=  \frac{\phi _{1}}{\phi _{2}}2p+p^{2}.
	\end{eqnarray*}%
	Thus, $
	\func{cov}\left( \bar{\mathcal{W}}_{p}\left( \phi _{1}\right) ,\bar{\mathcal{%
			W}}_{p}\left( \phi _{2}\right) \right) =\frac{\phi _{1}}{\phi _{2}}2p.
	$
	
	Since $\bar{\mathcal{W}}_{p}\left( \phi \right) $ can be represented by the
	sum of $p$ iid standard normals, the finite-dimensional convergence of the
	process $\mathsf{Q}_{p}\left( \phi \right) :=\left( \bar{\mathcal{W}}%
	_{p}\left( \phi \right) -p\right) /\sqrt{2p}$ to a centered gaussian
	process, say $\mathsf{Q}\left( \phi \right) $, with the  covariance
	kernel $
	E\mathsf{Q}\left( \phi _{1}\right) \mathsf{Q}\left( \phi _{2}\right) =\frac{%
		\min\left\{\phi _{1}, \phi _{2}\right\}}{\max\left\{\phi _{1}, \phi _{2}\right\}},
	$
	(as $p\rightarrow \infty $) is straightforward by the CLT and Cramer-Wold
	device. It is worth noting that $\mathsf{Q}_{p}\left( \phi \right) $ is
	invariant to $p$ up to the second moment. This is obvious for each $\phi $
	due to the centering and scaling but not for the covariance kernel.
	
	Next, we verify the uniform tightness of the processes. \cite{de1981crossing}
	derived the exact crossing probability for $\mathcal{\bar{W}}_{p}\left( \phi
	\right) $ in its equation (4) through an expansion and an upper bound for
	the tail probability on p.2205, which is given by 
	\begin{eqnarray*}
		&&P\left\{ \sup_{1<\phi <T}B_{p}\left( \phi \right) ^{\prime }B_{p}\left( \phi
		\right) /\phi >c\right\} \\
        &\leq& \frac{\left( c/2\right) ^{p/2}\exp \left(
			-c/2\right) }{\Psi \left( p/2\right) }\left( \log \left( T\right) \left( 1-
		\frac{p}{c}\right) +\frac{2}{c}+O\left( \frac{1}{c^{2}}\right) \right),
	\end{eqnarray*}
	where $\Psi \left( {}\right) $ is the Gamma function. Set $c=\bar{c}\sqrt{%
		2p}+p$ and then note that the upper bound is bounded for all $p$ and the the
	bound goes to zero as $\bar{c}\rightarrow \infty ,$ implying that the
	processes $\left( \mathcal{\bar{W}}_{p}\left( \phi \right) -p\right) /\sqrt{%
		2p}$ is uniformly tight. More specifically, note that $\frac{p}{c^{2}}=\frac{%
		p}{\bar{c}\sqrt{2p}+p}<1,$ and 
	\begin{eqnarray*}
		\frac{\left( c/2\right) ^{p/2}\exp \left( -c/2\right) }{\Psi \left(
			p/2\right) } &=&\frac{\left( \frac{\bar{c}\sqrt{2p}+p}{2}\right) ^{p/2}\exp
			\left( -\left( \bar{c}\sqrt{2p}+p\right) /2\right) }{\Psi \left(
			p/2\right) } \\
		&\sim &\frac{\left( \frac{\bar{c}\sqrt{2p}+p}{2}\right) ^{p/2}\exp \left(
			-\left( \bar{c}\sqrt{2p}+p\right) /2\right) }{\sqrt{2\pi \left( p/2\right) }%
			\left( p/2\right) ^{p/2}e^{-p/2}}
	\end{eqnarray*}%
	where the approximation of the Gamma function is due to Stirling's formula
	for large $p$'s. For each $\bar{c},$ after cancellation, the log
	transformation of the ratio (up to a constant) ends up with 
	\begin{eqnarray*}
		&&\frac{p}{2}\log \left( \frac{p+\bar{c}\sqrt{2p}}{2}\right) -\frac{\bar{c}%
			\sqrt{p}}{\sqrt{2}}-\frac{1}{2}\log \left( \frac{p}{2}\right) -\frac{p}{2}%
		\log \left( \frac{p}{2}\right) \\
		&=&\frac{p}{2}\left( \log \frac{p}{2}+\log \left( 1+\frac{\bar{c}\sqrt{2}}{%
			\sqrt{p}}\right) \right) -\frac{\bar{c}\sqrt{p}}{\sqrt{2}}-\left( 1+\frac{p}{%
			2}\right) \log \frac{p}{2} \\
		&=&-\bar{c}^{2}+\sum_{j=1}^{\infty }\frac{\bar{c}^{j+2}2^{\left( j+2\right)
				/2}}{\left( j+2\right) p^{j/2}}\left( -1\right) ^{j+1}-\log \frac{p}{2},
	\end{eqnarray*}%
	whose the supremum over $p$ is bounded and goes to zero as $\bar{c}%
	\rightarrow \infty .$ Thus, $\mathsf{Q}_{p}\left( \phi \right) \Rightarrow 
	\mathsf{Q}\left( \phi \right) $ over any compact set follows.
	
	Finally, standard algebra can show that the covariance kernel of $\mathsf{Q%
	}(\psi )$ with $\psi =\phi /(1+\phi )$ matches that of $\mathcal{Z}%
	(\psi )$.
\end{proof}


\section{Lemmas}\label{app:lemmas}
\begin{lemma}
\label{lemma:Mhat_Torm} Under the conditions of Theorem \ref%
{thm:null_Chow}, for all sufficiently large $T$, 
$\sup _{\psi \in \Psi }\left \Vert \hat {P}(\psi )\right \Vert=O_p(1)$ and $\sup
_{\psi \in \Psi }\left \Vert \hat {P}(\psi )^{-1}\right \Vert
=O_{p}\left(\mu_T^{-1}\right).$
\end{lemma}

\begin{proof}
 Note that 
$
\left\Vert \hat{P}(\psi )^{-1}\right\Vert \leq \left\Vert \hat{P}(\psi
)^{-1}\right\Vert \left\Vert \hat{P}(\psi )-P(\psi )\right\Vert
\left\Vert P(\psi )^{-1}\right\Vert +\left\Vert P(\psi )^{-1}\right\Vert
,
$
so 
$
\left\Vert \hat{P}(\psi )^{-1}\right\Vert \left( 1-\left\Vert \hat{P}%
(\psi )-P(\psi )\right\Vert \left\Vert P(\psi )^{-1}\right\Vert
\right) \leq \left\Vert P(\psi )^{-1}\right\Vert ,
$
by two applications of the triangle inequality. Taking limits of the last displayed
expression as $T\rightarrow \infty $ and using Assumption \ref{ass:M_diff},
the rate condition yields $\left\Vert \hat{P}(\psi )^{-1}\right\Vert =O_{p}(\mu_T^{-1})$.
Next, noting that $\left\Vert \hat{P}(\psi )\right\Vert \leq \left\Vert \hat{P}(\psi
)-P(\psi )\right\Vert +\left\Vert P(\psi )\right\Vert ,$
the lemma follows by using Assumption \ref{ass:M_diff}.
\end{proof}

\begin{lemma}
\label{lemma:B_eigs} \sloppy Under the conditions of Theorem \ref%
{thm:null_Chow}, $\sup_{\psi \in \Psi }\left\{ \underline{\text{eig}}\left( {F}(\psi
)\right) \right\} ^{-1} =O\left(\mu_T^{-1}\right)
\quad$ and $\quad
\sup_{\psi \in \Psi }\overline{\text{eig}}\left(
F(\psi )\right) =O\left(\mu_T^{-2}\right).$
\end{lemma}

\begin{proof}
Observe that because 
\begin{equation}
P(\psi )^{-1}=\left[ 
\begin{array}{cc}
(1-\psi )^{-1}P^{-1} & (1-\psi )^{-1}P^{-1} \\ 
(1-\psi )^{-1}P^{-1} & \left[ \psi (1-\psi )\right] ^{-1}P^{-1}%
\end{array}%
\right] ,  \label{M_gammainv}
\end{equation}%
we have 
\begin{equation}
F(\psi )^{-1}=\psi (1-\psi )P\Xi ^{-1}P.  \label{B_gamma1}
\end{equation} 
\sloppy Now, $\left \{\underline {\text{eig}}%
\left ({F}(\psi )\right )\right \}^{-1}=\overline {\text{eig}}\left ({F}%
(\psi )^{-1}\right )$, which, using (\ref{B_gamma1}), is bounded by 
$C\overline {\text{eig}}\left (P\Xi ^{-1}P\right )=C\left \Vert P\Xi
^{-1}P\right \Vert \leq C\overline {\text{eig}}\left (P\right )^2\underline {%
\mu }\left (\Xi \right )^{-1}=O(\mu_T^{-1}),$
uniformly on the compact $\Psi $, using Assumption \ref{ass:M_diff}$(ii)$. For the second part of the claim, because (\ref%
{B_gamma1}) implies $F(\psi )=\left [\psi (1-\psi
)\right
]^{-1}P^{-1}\Xi P^{-1}$, it follows similarly that $\overline {%
\text{eig} }\left (F(\psi )\right )$ is uniformly bounded by a constant times 
$\overline {\text{eig}}\left (P^{-1}\Xi P^{-1}\right )=\left \Vert
P^{-1}\Xi P^{-1}\right \Vert \leq \underline {\text{eig}}\left (P\right
)^{-2}\overline {\text{eig}}\left (\Xi \right )=O(\mu_T^{-2}).$
\end{proof}

\begin{lemma}
\label{lemma:Bhat_Torm} Under the conditions of Theorem \ref{thm:null_Chow}, $\sup_{\psi \in \Psi }\overline{\text{eig}} \left(\hat{F}(\psi )\right)
=O_p(\mu_T^{-2})$ and $\sup_{\psi \in \Psi }\overline{\text{eig}} \left( \hat{F}(\psi )^{-1}\right)
=O_{p}(\mu_T^{-1}).$
\end{lemma}

\begin{proof}
We show the second claim, the first following easily by the definition of $\hat{F}(\psi)$. First, define $\tilde{F}(\psi )=S\hat{P}(\psi )^{-1}\Xi (\psi )\hat{%
M}(\psi )^{-1}S^{\prime }$. We will use uniform bounds in the calculations
without explicitly mentioning this in each step to simplify notation.
Proceeding as in the proof of Lemma \ref{lemma:Mhat_Torm}, we can write 
\begin{align}
\left\Vert \hat{F}(\psi )^{-1}\right\Vert \left( 1-\left\Vert \hat{F}%
(\psi )-\tilde{F}(\psi )\right\Vert \right) & \leq \left\Vert \tilde{F}%
(\psi )^{-1}\right\Vert ,  \label{Bhat_til} \\
\left\Vert \tilde{F}(\psi )^{-1}\right\Vert \left( 1-\left\Vert \tilde{F}%
(\psi )-F(\psi )\right\Vert \right) & \leq \left\Vert F(\psi
)^{-1}\right\Vert .  \label{Bhat_B}
\end{align}%
Next, Lemma \ref{lemma:Mhat_Torm} implies 
\begin{equation}
\left\Vert \hat{F}(\psi )-\tilde{F}(\psi )\right\Vert \leq \left\Vert
S\right\Vert ^{2}\left\Vert \hat{P}(\psi )^{-1}\right\Vert ^{2}\left\Vert 
\hat{\Xi}(\psi )-\Xi (\psi )\right\Vert =O_{p}\left( \mu_T^{-2}\zeta_p \right)=o_p(1) .
\label{Bhat_til_diff}
\end{equation}%
 Next, $\tilde{F}(\psi )-F(\psi )$ equals $\left[ \hat{P}(\psi )^{-1}\Xi (\psi )\hat{P}(\psi )^{-1}-{P}%
(\psi )^{-1}\Xi (\psi ){P}(\psi )^{-1}\right] S^{\prime },$ which is
$
 S {P}(\psi )^{-1}\left[\left( \hat{P}(\psi )-{P}(\psi )\right) \hat{P}(\psi )^{-1}\Xi (\psi )+\Xi (\psi ){P}(\psi )^{-1}\left( \hat{P}(\psi
)-{P}(\psi )\right) \right]\hat{P}(\psi )^{-1}S^{\prime }. $ 
By this fact, Assumption \ref{ass:M_diff}, Lemmas \ref{lemma:Omega_hat_true} and \ref{lemma:Mhat_Torm}, we deduce that 
\begin{equation}
\left\Vert \tilde{F}(\psi )-F(\psi )\right\Vert =O_{p}\left( \mu_T^{-3}\tau_p  \right) =o_{p}(1).
\label{Bhat_diff}
\end{equation}%
The lemma now follows by taking limits of (\ref{Bhat_til}) and (\ref{Bhat_B}%
), and using (\ref{Bhat_til_diff}), (\ref{Bhat_diff}) and Lemma \ref%
{lemma:B_eigs}.
\end{proof}

\begin{lemma}
\label{lemma:deltahat}  \sloppy Under the conditions of Theorem \ref%
{thm:null_Chow}, $\sup_{\psi \in \Psi }\left\Vert \kappa -%
\hat{\kappa}(\psi )\right\Vert =O_{p}\left( \mu_T^{-1}\sqrt{p/T}\right) .$
\end{lemma}
\begin{proof}
Note that $\kappa -\hat{\kappa}(\psi )=\hat{P}(\psi
)^{-1}T^{-1}\sum_{t=1}^{T}x_{t}(\psi )e_{t}$ and that 
\begin{eqnarray*}
\left\Vert \kappa -\hat{\kappa}(\psi )\right\Vert ^{2} &\leq&
\left\Vert \hat{P}(\psi )^{-1}\right\Vert ^{2}T^{-2}\left\Vert
\sum_{t=1}^{T}x_{t}(\psi )e_{t}\right\Vert ^{2}=O_{p}\left(\mu_T^{-2}\right)
T^{-2}\left\Vert \sum_{t=1}^{T}x_{t}(\psi )e_{t}\right\Vert ^{2} \\
&\leq&O_{p}\left(\mu_T^{-2}\right)\left( T^{-2}\left\Vert \sum_{t=1}^{T}x_{t}(\psi )\nu
_{t}\right\Vert ^{2}+T^{-2}\left\Vert X(\psi )^{\prime }r\right\Vert
^{2}\right) ,
\end{eqnarray*}%
uniformly in $\psi $, by Lemma \ref{lemma:Mhat_Torm}. 

\sloppy Next, $E\left(
T^{-2}\left\Vert \sum_{t=1}^{T}x_{t}(\psi )\nu _{t}\right\Vert
^{2}\right)=E\left( T^{-2}\sum_{s,t=1}^{T}x_{t}^{\prime }(\psi )x_{s}(\psi
)\nu _{s}\nu _{t}\right)$, 
which equals
\begin{equation}
T^{-2}\sum_{t=1}^{T}E\left\Vert x_{t}(\psi )\right\Vert ^{2}\sigma
_{t}^{2}+2T^{-2}\sum_{s<t}E\left(x_{t}^{\prime }(\psi )x_{s}(\psi )E\left(
\nu _{s}E\left( \nu _{t}|\nu _{r},r<t\right) \right)\right)
=O_{p}\left( p/T\right) ,  \label{lemma2_first_expec}
\end{equation}%
by Assumption \ref{ass:errors}. Finally, $T^{-2}\left\Vert X(\psi )^{\prime }r\right\Vert ^{2}\leq T^{-2}\left\Vert
X(\psi )\right\Vert ^{2}\left\Vert r\right\Vert ^{2}=\overline{\text{eig} }
\left( \hat{P}(\psi )\right) T^{-1}\left\Vert r\right\Vert
^{2}=O_{p}\left( 1/T\right)$, by Assumption \ref{ass:aprx0} and Lemma \ref{lemma:Mhat_Torm}. Therefore, 
\begin{equation}
\sup_{\psi \in \Psi }\left\Vert \kappa -\hat{\kappa}(\psi )\right\Vert
=O_{p}\left(\mu_T^{-1} \sqrt{p}/\sqrt{T}\right) ,  \label{delta_deltahat_order}
\end{equation}%
by Markov's inequality.
\end{proof}

\begin{lemma}
        \label{lemma:Omega_hat_true} \sloppy Write $\tilde{\Xi}(\psi )=T^{-1}\sum_{t=1}^{T}x_{t}(\psi
)x_{t}^{\prime }(\psi )\nu _{t}^{2}$. Under Assumptions \ref{ass:errors}-\ref{ass:M_diff}, $\sup_{\psi \in \Psi }\left\Vert \hat{\Xi}(\psi )-\tilde{\Xi}
        (\psi )\right\Vert  =O_{p}\left( {p^{3}}/{T\mu_T^{2}} \right)$ and $\sup_{\psi \in \Psi }\left\Vert \tilde{\Xi}(\psi )-\bar{\Xi}
        (\psi )\right\Vert  =O_{p}\left( {p}/{\sqrt{T}}\right) .$
 \end{lemma}

\begin{proof}
	\sloppy We start by proving the first claim. The matrix inside the norm on the LHS can be
	decomposed as $\sum_{i=1}^{5}\upalpha_{i}(\psi )$, with $\upalpha_{1}(\psi )=T^{-1}\sum_{t=1}^{T}x_{t}(\psi )x_{t}^{\prime }(\psi ) 
	\left[ x_{t}^{\prime }(\psi )\left( \kappa -\hat{\kappa}(\psi )\right)
	\right] ^{2}, 
	\upalpha_{2}(\psi ) =T^{-1}\sum_{t=1}^{T}x_{t}(\psi )x_{t}^{\prime }(\psi
	)r_{t}^{2}, 
	\upalpha_{3}(\psi ) =2T^{-1}\sum_{t=1}^{T}x_{t}(\psi )x_{t}^{\prime }(\psi ) 
	\left[ x_{t}^{\prime }(\psi )\left( \kappa -\hat{\kappa}(\psi )\right) 
	\right] \nu _{t}, 
	\upalpha_{4}(\psi ) =2T^{-1}\sum_{t=1}^{T}x_{t}(\psi )x_{t}^{\prime }(\psi ) 
	\left[ x_{t}^{\prime }(\psi )\left( \kappa -\hat{\kappa}(\psi )\right) 
	\right] r_{t},
	\upalpha_{5}(\psi ) =2T^{-1}\sum_{t=1}^{T}x_{t}(\psi )x_{t}^{\prime }(\psi
	)r_{t}\nu _{t}.$ Recall Lemma (\ref{lemma:deltahat}) for $\sup_{\psi \in \Psi }\left\Vert
	\kappa -\hat{\kappa}(\psi )\right\Vert =O_{p}\left(\mu_T^{-1} \sqrt{p/T}\right) .$
	Now, since the maximum eigenvalue of a non-negative definite symmetric
	matrix is less than equal to the trace, $\left\Vert \upalpha_{1}(\psi )\right\Vert \leq T^{-1}\sum_{t=1}^{T}\left\Vert
	x_{t}(\psi )\right\Vert ^{4}\left\Vert \kappa -\hat{\kappa}
	(\psi )\right\Vert ^2  = O_{p}\left(\mu_T^{-2} p^{2}\right) O_{p}\left( p/T\right) ,$
	uniformly in $\psi $, by the fact that $\sup_{t,j}Ex_{tj}^{4}<\infty $ and Lemma
	(\ref{lemma:deltahat}). In a similar fashion, $
		E\left\Vert \upalpha_{2}(\psi )\right\Vert \leq
		2ET^{-1}\sum_{t=1}^{T}x_{t}^{\prime }x_{t}r_{t}^{2}\leq 2\left( E\left(
		x_{t}^{\prime }x_{t}\right) ^{2}Er_{t}^{4}\right) ^{1/2}=O\left( \mu_T^{-2}p/\sqrt{T}%
		\right) .$
	Similarly and using the fact that $E\left( \left\vert
	\nu _{t}\right\vert |x_{t}\right) \leq \sqrt{E\left( \nu
		_{t}^{2}|x_{t}\right) }=O\left( 1\right) ,$ we obtain $\left\Vert \upalpha_{3}(\psi )\right\Vert \leq 4T^{-1}\sum_{t=1}^{T}\left(
		x_{t}^{\prime }x_{t}\right) ^{2}\left\vert \nu _{t}\right\vert
		\left\Vert \kappa -\hat{\kappa}(\psi )\right\Vert ^{2}=O_{p}\left(
		\mu_T^{-2}p^{3}/T\right)$,
		$\left\Vert \upalpha_{4}(\psi )\right\Vert \leq 4\left\Vert \kappa -\hat{\kappa}(\psi )\right\Vert \left(
		T^{-1}\sum_{t=1}^{T}\left( x_{t}^{\prime }x_{t}\right) ^{2}\right)
		^{3/4}\left( T^{-1}\sum_{t=1}^{T}r_{t}^{4}\right) ^{1/4} 
		= O_p\left( \sqrt{{p}/{T}}p^{3/2}{\mu_T^{-1}}{T^{-1/4}}\right)$
		and $\left\Vert \upalpha_{5}(\psi )\right\Vert  =2T^{-1}\sum_{t=1}^{T}\left(
		x_{t}^{\prime }x_{t}\right) \left\vert r_{t}\nu _{t}\right\vert
		=O_{p}\left( p/\sqrt{T}\right) ,$ all uniformly in $\psi\in\Psi $. Thus the first claim is established.
	
	To show the second claim, let $x_{it}$, $i=1,\ldots ,p$, be a typical
	element of $x_{t}$. Then any element of $\tilde{\Xi}(\psi )-\bar{\Xi}%
	(\psi )$ is of the form $T^{-1}\sum_{t=1}^{T}x_{it}(\psi )x_{jt}(\psi
	)\left( \nu _{t}^{2}-\sigma _{t}^{2}\right) $, $i,j=1,\ldots ,p$,
	and $\nu _{t}^{2}-\sigma _{t}^{2}$ is an MDS by construction. Thus,
	it has mean zero and variance $
	T^{-2}\sum_{t=1}^{T}Ex_{it}^{2}(\psi )x_{jt}^{2}(\psi )E\left( \left(
	\nu _{t}^{2}-\sigma _{t}^{2}\right) ^{2}|\mathfrak{G}_{t-1}\right)
	=O_{p}\left( T^{-1}\right)$, 
	by Assumption \ref{ass:errors} and the boundedness of $Ex_{it}^{4}$. Thus, 
	$
	E\left\Vert \tilde{\Xi}(\psi )-\bar{\Xi}(\psi )\right\Vert
	^{2}=O\left( p^{2}/T\right) ,
	$ and the claim follows by Markov's inequality.
\end{proof}
\end{appendix}
\begin{acks}[Acknowledgments]
Corresponding author: Seo. Gupta is also affiliated with Department of Economics, University of Essex, Wivenhoe Park, Colchester, CO4 3SQ, UK.
\end{acks}


\bibliographystyle{imsart-nameyear} 
\bibliography{AGmasterref}       

\end{document}